\documentclass[journal]{IEEEtran}
\usepackage{amsmath,amssymb,amsfonts}
\usepackage{amsthm}
\usepackage{graphicx}
\usepackage{url}
\usepackage{cite}
\usepackage{hyperref}
\usepackage[nolist]{acronym}
\usepackage[caption=false,font=normalsize,labelfont=sf,textfont=sf]{subfig}
\usepackage{tikz}
\usetikzlibrary{arrows.meta,positioning,backgrounds,calc,fit}
\definecolor{phys}{HTML}{1F6FEB}
\definecolor{dig}{HTML}{0FA3A3}
\definecolor{comb}{HTML}{7B2FBF}
\definecolor{inkgray}{HTML}{444444}
\definecolor{los}{HTML}{1B9E5A}
\definecolor{nlosA}{HTML}{E8590C}
\definecolor{nlosB}{HTML}{7B2FBF}
\tikzset{
 bs/.pic={
   \draw[very thick,inkgray] (-0.22,0)--(0,1.0)--(0.22,0);
   \draw[very thick,inkgray] (-0.13,0.45)--(0.13,0.45);
   \draw[very thick,inkgray] (0,1.0)--(0,1.28);
   \draw[inkgray,thin] (0,1.28)++(140:0.18) arc (140:40:0.18);
   \draw[inkgray,thin] (0,1.28)++(140:0.30) arc (140:40:0.30);
 },
 ue/.pic={
   \draw[very thick,inkgray,fill=white,rounded corners=1.4pt] (-0.17,-0.3) rectangle (0.17,0.34);
   \draw[inkgray,thin,fill=phys!12,rounded corners=0.6pt] (-0.12,-0.18) rectangle (0.12,0.26);
 },
 bldg/.pic={
   \draw[thick,inkgray,fill=inkgray!10] (-0.9,0) rectangle (0.9,0.95);
   \foreach \x in {-0.6,-0.2,0.2,0.6}{\foreach \y in {0.2,0.5,0.78}{
     \fill[inkgray!35] (\x-0.09,\y-0.07) rectangle (\x+0.09,\y+0.07);}}
 },
}
\newcommand{\caricon}[2]{%
  \begin{scope}[shift={#1}]
   \draw[thick,inkgray,fill=#2,rounded corners=2pt] (-0.5,0.02) rectangle (0.5,0.2);
   \draw[thick,inkgray,fill=#2] (-0.26,0.2)--(0.2,0.2)--(0.12,0.4)--(-0.18,0.4)--cycle;
   \fill[inkgray] (-0.28,0.02) circle (0.09);\fill[white] (-0.28,0.02) circle (0.038);
   \fill[inkgray] (0.28,0.02) circle (0.09); \fill[white] (0.28,0.02) circle (0.038);
  \end{scope}}

\newtheorem{assumption}{Assumption}
\newtheorem{proposition}{Proposition}
\newtheorem{lemma}{Lemma}
\newtheorem{corollary}{Corollary}

\newtheorem{theorem}{Theorem}

\DeclareMathOperator*{\argmin}{arg\,min}
\DeclareMathOperator*{\argmax}{arg\,max}
\newcommand{\E}{\mathbb{E}}
\newcommand{\CN}{\mathcal{CN}}

\newcommand{\hPW}{\hat{h}_{\mathrm{PW}}}
\newcommand{\hPD}{\hat{h}_{\mathrm{P+D}}}
\newcommand{\hDT}{\hat{h}_{\mathrm{DT}}}
\newcommand{\sDT}{\sigma_{\mathrm{DT}}^{2}}

\begin{document}
\begin{acronym}[MMSE]
	\acro{PW}{physical world}
	\acro{MLE}{maximum likelihood estimator}
    \acro{CSI}{channel state information}
    \acro{DT}{digital twin}
    \acro{BS}{base station}
    \acro{UE}{user equipment}
    \acro{MISO}{multiple input single output}
    \acro{OFDM}{orthogonal frequency division multiplexing}
    \acro{ML}{maximum likelihood}
    \acro{MAP}{maximum a posteriori}
    \acro{BLUE}{best linear unbiased estimator}
    \acro{CRB}{Cram\'er--Rao bound}
    \acro{MMSE}{minimum mean square error}
    \acro{MSE}{mean square error}
    \acro{SNR}{signal to noise ratio}
    \acro{CLT}{central limit theorem}
    \acro{MIMO}{multiple-input and multiple-output}
    \acro{RIS}{reconfigurable intelligent surfaces}
    \acro{ISAC}{integrated sensing and communications}
\end{acronym}

\title{How Much Training is Needed with a Digital Twin?}

\author{Ahmad Bazzi,  Marwa Chafii 
\thanks{
Ahmad Bazzi and Marwa Chafii are with Engineering Division, New York University (NYU) Abu Dhabi, 129188, UAE and NYU WIRELESS,
NYU Tandon School of Engineering, Brooklyn, 11201, NY, USA (email: \href{ahmad.bazzi@nyu.edu}{ahmad.bazzi@nyu.edu}, \href{marwa.chafii@nyu.edu}{marwa.chafii@nyu.edu}).}
\thanks{Manuscript received xxx}}

\markboth{IEEE Transactions on WIRELESS COMMUNICATIONS}%
{Bazzi: How Many Pilots Is a Digital Twin Worth?}

\maketitle

\begin{abstract}
The following paper addresses how much pilot training is needed when a digital twin (DT) of the wireless radio channel is available to aid a wireless communication system with a channel estimation task. The DT of a wireless channel is widely expected to reduce the pilot overhead of channel estimation, following the informal rule that \emph{``the more accurate the twin, the fewer pilots are needed.''} This trade-off, however, has only ever been demonstrated empirically and never quantified. We close this gap by treating the DT as a complementary measurement of the channel that the receiver fuses with its pilot observations in the physical world. Consequently, fusing the physical and digital worlds through the best linear unbiased estimator, we derive a DT-aided Cram\'er-Rao bound, and from it a \emph{pilot-equivalence law} that converts DT fidelity into an equivalent number of training symbols. For a biased twin unknown to the estimator, we obtain the exact mismatch threshold beyond which trusting the DT is worse than ignoring it. 
We quantify how much training is needed with the DT to attain a desired mean square error on channel estimation. 
Particular cases are discussed to tell when training in the physical world can be completely bypassed. 
We finally translate these results into a block-fading achievable rate whose optimal training length is the unique root of a single equation, and identify the DT fidelity above which pilot training can be dispensed with altogether.
Extensive numerical results corroborate closed-form expression and reveal that the value of a DT is largest at finite signal-to-noise ratio and vanishes in both the low- and high-SNR limits.
\end{abstract}


\section{Introduction}\label{sec:intro}
\IEEEPARstart{O}{{riginating}} in the manufacturing and aerospace communities two decades ago~\cite{10255711}, and in particular within the United States Air Force Research Laboratory, a \ac{DT} is a virtual replica of a physical entity \cite{9854866,10198573,11575744} that is kept synchronized with it through a continuous exchange of measurements and predictions~\cite{10234596,9994764}. In wireless communications, the \ac{DT} has become a pillar of the $6$G vision, which contributes in maintaining a live model of the radio environment, traffic, and hardware, and serving as a risk free \emph{sandbox} in which the network can be designed and optimized without experimenting on the live system~\cite{10628026,10012285,10417097}. The \ac{DT} can help $6$G wireless technologies enable new applications such as autonomous driving, indoor localization \cite{li2026xl}, and environmental monitoring and imaging \cite{11519592,zhang2024new,11400690,11173662,10475383}, in addition to \ac{ISAC} \cite{10285442}. 
Moreover, a \ac{DT} can also be used as a simulation environment that goes beyond \ac{ISAC} and can be extended for more complex tasks with different multimodal wireless networks \cite{11534549}, where communications and radio-frequency are complemented with additional modalities, such as cameras \cite{fan2026heterogeneous}.
However, the usefulness of any given \ac{DT} is governed by a single quantity, which is the \ac{DT}'s \emph{fidelity}, namely how faithfully the virtual model tracks the physical one. This has been formalized both as model uncertainty induced by limited data~\cite{10234596,9812625} and, information theoretically, as a joint function of model quality and the information exchanged between the physical and digital worlds~\cite{9994764}.

At the network level, \acp{DT} support zero touch $6$G management and resource orchestration \cite{10716598} across edge, aerial, enable security \cite{zhang2026ambsentry,lei2026resourceallocationsecuredualuavassisted}, vehicular, and industrial systems~\cite{10638530,10930648,10234388,10234624,10070572,10061692,10972174,10628034}, integrated sensing and communications \cite{10255711,10836879,mollahosseini2025integrated,zhang2025integrated} and edge computing links~\cite{10234427,10319784,10371218,10944626,10255711}, and are increasingly driven by generative artificial intelligence~\cite{10628026,10417097,10012285}, with recent analysis bounding the deployed performance of a \ac{DT} trained policy by a formal twin to reality discrepancy~\cite{11215841,10253478}. \emph{Our interest, however, lies one layer down, which is within the \ac{DT} of the wireless radio channel.}

The \ac{DT} of a wireless radio channel can be built by importing a $3$ dimensional model of some site-specific environment, assigning electromagnetic materials to its surfaces, and possibly synthesize the channel through possible ray tracing approaches~\cite{10234421,10990226,11195801}. On this basis, the \acp{DT} can fuse ray tracing with learning for multi-spectrum propagation modeling~\cite{10234421}, map user positions to statistical \ac{CSI} via some diffusion model~\cite{10906057}, counter channel aging by generative data augmentation~\cite{10417075}, track \ac{RIS} channels \cite{11162110,10854667} with physics informed neural networks~\cite{11016226,11457984}, and place ray tracing directly in the vehicular simulation process~\cite{11513690}. \emph{Since the rendered channel rests on an imperfect geometry and material description, however, a \ac{DT} never matches reality exactly; this residual sim-to-real gap, embedded in geometry and material mismatch, is one fundamental challenge in what caps fidelity} ~\cite{10234596,9994764}.

But even with this gap, the \ac{DT} has proven to be a powerful tool for lowering  overhead of pilot based channel acquisition, where the overhead reduction will be most useful in the large antenna array regime and millimeter wave bands envisioned for $6$G, because pilots are expensive due to a multitude of reasons: $(i)$ pilots will create overhead in learning the channel, and training too much will eventually lower the achievable rate when communicating within a fixed coherence interval, $(ii)$ reusing pilots across cells introduces the so-called pilot contamination, and $(iii)$ and the cost of beam sweeping grows with the number of antenna elements~\cite{10990226}, which becomes more dramatic in $6$G as massive \ac{MIMO} is expected to be used \cite{11573792}. \emph{A first body of work exploits the fact that a site-specific \ac{DT} can generate its own channel data.} Such synthetic \ac{CSI} has been used to train feedback and beam management networks, which then require far less real world data than models trained from measurements alone~\cite{11195801,11006050,11125862,11345189}. \emph{A second body of work uses the \ac{DT} rendered channel directly, which is more optimistic as more confidence is placed on the \ac{DT} itself.} Here, the channel produced by the \ac{DT} drives the precoder through some over-the-air channel impulse response inference~\cite{10990226}, and so statistical \ac{CSI} can be obtained from the user positions without transmitting any pilots at all~\cite{10906057}.

Across the literature, an informal statement reads as follows: \textit{``the more accurate the \ac{DT}, the fewer pilots/training are needed.''}~\cite{10906057,10990226,11195801,10942924,jiang2023digital}. \textbf{It is, however, only ever demonstrated empirically for a particular system. It has never been quantified, to the best of the author's knowledge.} For instance, the \ac{DT} approach in \cite{11195801} contains the relevant site-specific information, which was stated to require less feedback overhead to achieve the same or even better performance than a generic dataset. Furthermore, \cite{10906057} does not require any pilots so that their framework can generate statistical \acp{CSI} from a large number of positions, however achieving satisfactory results. Also, \cite{jiang2023digital} argues that if the digital replica is adequately accurate, the real world channel acquisition can even be bypassed. Moreover, \cite{10942924} infers that a \ac{DT} can predict the information about the wireless channel in the real environment even by eliminating the channel acquisition overhead. Also, \cite{10990226} provides numerical results showing the accuracy of the proposed channel \ac{DT} without \ac{CSI} acquisition. Meanwhile, \cite{arnold2024vision} highlighted that a key challenge in the application of \ac{DT} technology is the channel simulators requirements on the accuracy of the 3D \ac{DT}. In retrospect, one may even wonder, \emph{how many pilots is a \ac{DT} worth? how does that worth scale with some \ac{DT}'s uncertainty about the environment, and when does trusting a \ac{DT} become worse than using no \ac{DT} at all?} or even more broadly, \emph{how much training is needed with a digital twin ?}

Classical training theory quantifies pilots against \ac{SNR} and/or number of antennas in the \ac{MIMO} case~\cite{hassibi2003howmuch}, while existing notions of \ac{DT} fidelity~\cite{9994764,11215841} do not translate into pilots. We intend to close this gap by modeling the \ac{DT} as supplemental information on top of the training observations in the physical world, which serves an additional  measurement of the wireless channel. In the following, we highlight the novel contributions and shed light on some important insights appearing in this manuscript. In the following, we highlight the novel contributions and shed light on some important insights appearing in this manuscript. To this end, we summarize our contributions as follows.

\begin{itemize}
\item \textbf{\ac{DT} measurement model.} We treat the \ac{DT} as a complementary measurement of the wireless channel, available to the receiver on top of the training observations acquired in the \ac{PW}. We then establish the conditions under which this measurement admits a Gaussian model, by invoking the Lindeberg-Feller \ac{CLT} over the  path errors of the \ac{DT}. We further show that both moments of the Gaussian model are dictated by the electromagnetic properties the \ac{DT} assigns to the surfaces. In particular, we simulate the \ac{DT} error from the Fresnel coefficients and the full polarimetric path gain over a deployment ensemble, and we demonstrate that a wrong material assumption increases the variance, whereas a slowly varying path phase turns into a bias that turns out to be frequency-dependent.

\item \textbf{Pilot equivalence of a \ac{DT} and its mismatch threshold.} We fuse the \ac{PW} and the digital world through the \ac{BLUE} estimator and derive a \ac{DT} aided \ac{CRB} that interpolates continuously between pilot only estimation, obtained for an uninformative \ac{DT}, and genie aided \ac{CSI}, obtained for a perfect \ac{DT}. Thanks to this expression, we derive a pilot equivalence law that converts the precision of the \ac{DT} into an equivalent number of training symbols, thereby giving a measure for \ac{DT} fidelity, together with the training length needed to reach a target \ac{MSE}. For the case of a biased \ac{DT} that is unknown to the estimator, we obtain the exact mismatch threshold beyond which trusting the \ac{DT} is  worse than ignoring it. Beyond that threshold, the \ac{DT} is worth a negative number of pilots, an effect we refer to as unlearning of the \ac{PW}.

\item \textbf{Achievable rate and optimal training length.} We translate the above estimation results into an achievable rate over a block fading channel, for both an unbiased and a biased \ac{DT}, and we derive the training length that maximizes this rate as the unique root of a single equation. Thanks to these expressions, we recover the pilot-only and the genie-aided rates as limiting cases, and we identify the \ac{DT} fidelity above which training can be dispensed with altogether.
\end{itemize}

Finally, we present extensive numerical results that corroborate every closed form expression and unveil some important insights, i.e. 
\begin{itemize}
\item An incorrect material assumption within the \ac{DT} can increase the \ac{DT} uncertainty (or the \ac{DT} variance) by a factor of $3.8$ when the surfaces are rendered as a light dielectric instead of a near metal, while the very same error turns into a bias at lower frequencies.

\item The number of pilots a \ac{DT} replaces decays inversely with the operating \ac{SNR} in the \ac{PW} for the unbiased case and is further reduced through a factor that depends on the bias as well as the amount of training used for the biased case, according to the pilot-equivalence law.

\item Training in the \ac{PW} shrinks as the \ac{DT} becomes more accurate, to the point where no training is needed at all, and it increases once the \ac{DT} is biased or less accurate. For a coherence block of $30$ symbols and a \ac{DT} variance of $0.05$, an unbiased \ac{DT} attains $3\,\mathrm{bits/s/Hz}$ with a single training symbol, whereas a bias of $0.3$ needs $3$ training symbols to deliver $2.75\,\mathrm{bits/s/Hz}$, which is exactly the best rate attained with no \ac{DT} at all, and a bias of $0.55$ needs about $4$ training symbols while delivering only $2.5\,\mathrm{bits/s/Hz}$. An unbiased \ac{DT} operating at $10 \mathrm{dB}$ requires no training at all to achieve a channel \ac{MSE} error of $10^{-2}$ as long as the \ac{DT} variance does not exceed $10^{-3}$.
\item The value of a \ac{DT} in terms of achievable rate is maximized (when compared to the achievable rate of no \ac{DT}) at a moderate \ac{SNR} and vanishes in both the low and the high \ac{SNR} limits. For instance, when the variance of the \ac{DT} is $0.05$, the optimal \ac{SNR} for maximum value of a \ac{DT} in terms of achievable rate occurs at $2.5$, and that \ac{SNR} decreases with increasing bias. For a highly biased \ac{DT}, e.g. $0.35$, there exists \ac{SNR} regions where the \ac{DT} is more harmful over a mid \ac{SNR} range.
\end{itemize}

\textbf{Organization}: The remainder of this paper is organized as follows. Section~\ref{sec:model} introduces the system model, casting the \ac{DT} as a complementary Gaussian measurement of the channel. Section~\ref{sec:analysis} fuses this measurement with the physical-world pilots through the best linear unbiased estimator and derives the DT-aided \ac{CRB}, the pilot-equivalence law, and, for a biased twin, the mismatch threshold and the associated unlearning regime. Section~\ref{sec:cap} translates these estimation results into a block-fading achievable rate and obtains the optimal training length in closed form. Section~\ref{sec:sim} presents numerical results that corroborate the closed-form expressions and quantify the value of a DT across signal-to-noise ratio, \ac{DT}variance, as well as the \ac{DT} bias. Finally, Section~\ref{sec:concl} concludes the paper.

\textbf{Notation}: Lowercase letters denote scalars and $\E\{\cdot\}$ is the statistical expectation. For a complex number $z\in\mathbb{C}$, $|z|$ is its magnitude. $\CN(\mu,\sigma^2)$ denotes the circularly symmetric complex Gaussian distribution with mean $\mu$ and variance $\sigma^2$. The operator $(x)^{+}\triangleq\max(x,0)$.

\section{System Model}\label{sec:model}
\begin{figure*}[!t]
\centering
\resizebox{\textwidth}{!}{%
\begin{tikzpicture}[>=Stealth,font=\small,
  losS/.style={los,line width=1.3pt}, naS/.style={nlosA,line width=1.1pt},
  nbS/.style={nlosB,line width=1.1pt}, dot/.style={circle,fill=inkgray,inner sep=1.1pt}]
\fill[phys!5,rounded corners=5pt] (-0.15,0.15) rectangle (6.65,4.75);
\fill[dig!6,rounded corners=5pt]  (7.75,0.15) rectangle (14.55,4.75);
\pic at (0.7,2.25){bs}; \node[font=\footnotesize\bfseries] at (0.7,1.9){BS};
\pic at (6.0,2.05){ue}; \node[font=\footnotesize\bfseries] at (6.0,1.5){UE};
\pic at (3.4,3.05){bldg};
\caricon{(3.4,0.95)}{phys!25}\node[font=\footnotesize] at (3.4,0.45){car};
\draw[->,inkgray,thin] (4.05,1.0)--(4.5,1.0);
\coordinate (S) at (0.86,2.62);\coordinate (U) at (5.85,2.08);
\draw[losS] (S)--(U);
\draw[naS] (S)--(3.4,3.05)--(U);\node[dot] at (3.4,3.05){};
\draw[nbS] (S)--(3.4,1.35)--(U);\node[dot] at (3.4,1.35){};
\node[los, fill=phys!5,inner sep=1pt,font=\scriptsize] at (4.7,2.18){LoS};
\node[nlosA,font=\scriptsize] at (1.85,3.02){NLoS$_1$};
\node[nlosB,font=\scriptsize] at (1.85,1.6){NLoS$_2$};
\node[anchor=west,text=phys,font=\bfseries] at (-0.1,4.55){Physical world};
\begin{scope}[xshift=7.9cm]
\pic at (0.7,2.25){bs}; \node[font=\footnotesize\bfseries] at (0.7,1.9){BS};
\pic at (6.0,2.05){ue}; \node[font=\footnotesize\bfseries] at (6.0,1.5){UE};
\pic at (3.4,3.05){bldg};
\caricon{(3.4,0.95)}{inkgray!8}
\caricon{(3.95,0.95)}{dig!18}
\draw[<->,inkgray,thin] (3.4,0.6)--(3.95,0.6);\node[font=\scriptsize] at (3.68,0.38){};
\coordinate (S2) at (0.86,2.62);\coordinate (U2) at (5.85,2.08);
\draw[losS,dashed] (S2)--(U2);
\draw[naS,dashed] (S2)--(3.4,3.05)--(U2);\node[dot] at (3.4,3.05){};
\draw[nbS,dashed] (S2)--(3.95,1.35)--(U2);\node[dot] at (3.95,1.35){};
\node[los, fill=dig!6,inner sep=1pt,font=\scriptsize] at (4.7,2.18){LoS};
\node[anchor=west,text=dig,font=\bfseries] at (-0.1,4.55){Digital world};
\node[align=center,font=\scriptsize,text=dig] at (3.0,0.0){material mismatch $\Rightarrow (b,\sDT)$};
\end{scope}
\draw[-{Stealth[length=3mm]},line width=1.4pt,dig] (6.75,2.5)--(7.7,2.5)
  node[midway,above,font=\scriptsize,text=dig,align=center]{ray\\tracing};
\end{tikzpicture}}
\caption{The \ac{PW} propagation and its ray traced replica in the digital world, where a LoS path and two NLoS paths, off a building and off a passing car, link the \ac{BS} to the \ac{UE}. The \ac{DT} reproduces the dynamic scene only approximately and assigns the surfaces materials that differ from the true ones, which perturbs the path gains, and gives rise to the bias $b$ and the variance $\sDT$ in \eqref{eq:prior}.}
\label{fig:scene}
\end{figure*}
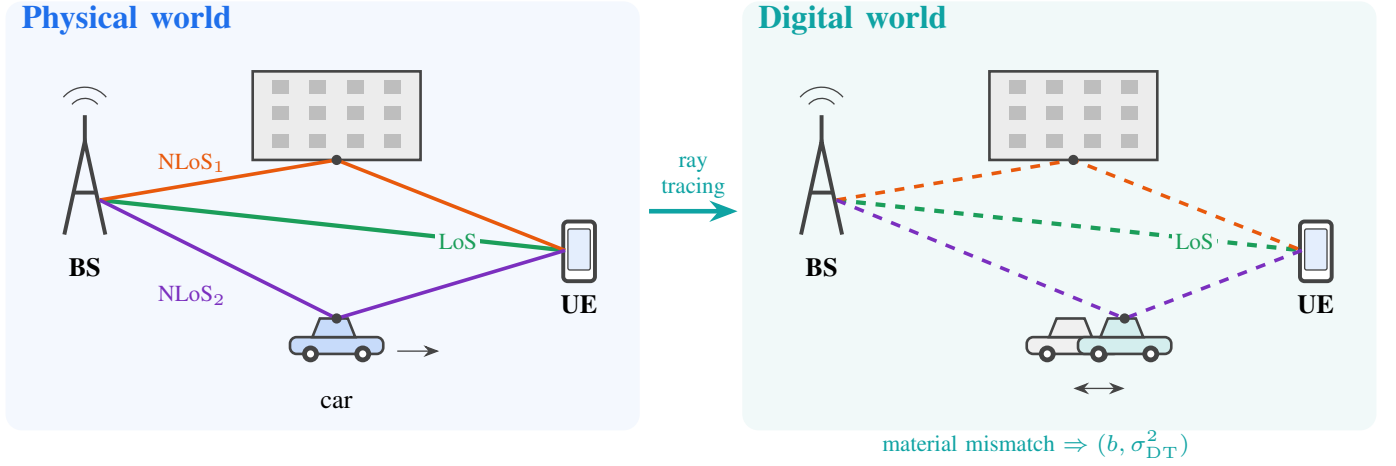

We analyze downlink transmission over a narrowband subcarrier within a \ac{MISO} and \ac{OFDM} system. 
The receiver intends to estimate the channel coefficient $h_0 \in \mathbb{C}$ that exists between the \ac{BS} and \ac{UE}. Within the subcarrier of interest, this channel is represented as the superposition of $L$ multipath components,
\begin{equation}
h_{0} = \sum\nolimits_{\ell=1}^{L} \alpha_{\ell}\,e^{-j2\pi f\tau_{\ell}},
\label{eq:multipath}
\end{equation}
where path $\ell$ has complex gain $\alpha_{\ell}$ and delay $\tau_{\ell}$ set by
the physical world geometry and the electromagnetic properties of the scatterers within the physical environment, and $f$ represents frequency.
Following the block fading training model of \cite{hassibi2003howmuch}, within a
channel coherence interval of $T$ symbols the \ac{BS} transmits $T_{\tau}\le T$
unit power training symbols (pilots), leaving $T-T_{\tau}$ symbols for data. For simplicity's sake, we assume an all-ones pilot frame. The
\ac{UE} observes noisy samples of the channel $h_{0}$,
\begin{equation}
y_{i} = h_{0} + n_{i}, 
\label{eq:pilot}
\end{equation}
where $ n_{i}\sim\CN(0,\sigma_{n}^{2})$ is additive white Gaussian noise over all the subcarriers, which is considered to be white and independent over the subcarriers, whose noise power is $\sigma_n^2$. The \ac{SNR} is defined as $\rho \triangleq 1/\sigma_{n}^{2}$.

Now let us assume a \ac{DT} intends to approximate the \ac{PW} geometry, including the material models. Moreover, a \ac{DT} attempts to reconstruct the superposition \eqref{eq:multipath} from its \emph{model} of the environment, where it generates a gain $\hat\alpha_{\ell}$ and delay $\hat\tau_{\ell}$ for each path it resolves over a path set $\hat{\mathcal{L}}$. For instance, the \ac{DT} can be a ray tracer which generates
\begin{equation}
\hDT = \sum\nolimits_{\ell\in\hat{\mathcal{L}}} \hat\alpha_{\ell}\,e^{-j2\pi f\hat\tau_{\ell}}.
\label{eq:twinpaths}
\end{equation}

\begin{figure*}[!t]
\centering
\subfloat[Error impact over different assumed materials.]{\includegraphics[height=2.25in,width=3.5in]{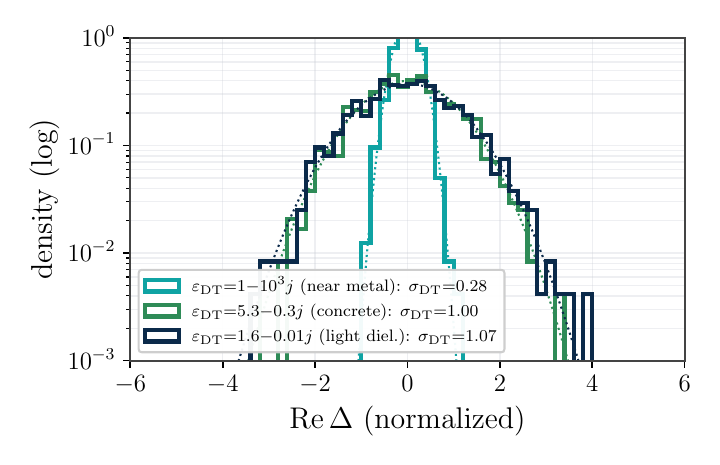}%
\label{fig:clt-a}}
\hfil
\subfloat[Error impact over carrier frequencies.]{\includegraphics[height=2.25in,width=3.5in]{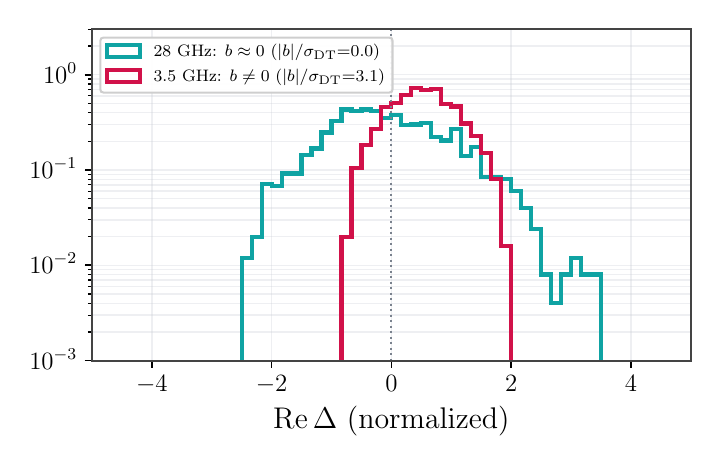}%
\label{fig:clt-b}}
\caption{Histogram of the \ac{DT} error $\Delta$ in a metal-walled room modelled by the \ac{DT} as the wrong material: (a) its variance $\sDT$ grows as the assumed material moves further from metal while staying zero mean, and (b) with the material fixed at concrete the same error stays zero mean at a high carrier ($28\,\mathrm{GHz}$) but adds coherently into a bias $b$ at a low carrier ($3.5\,\mathrm{GHz}$), realizing the model in \eqref{eq:prior}.}
\label{fig:clt}
\end{figure*}

For a single specular reflection off a planar surface of complex relative permittivity $\varepsilon_{r}$ at incidence angle $\theta$, the reflected field depends on the incident field and the Fresnel coefficient \cite{orfanidis}, which in turn fundamentally depend on the transverse electric and magnetic fields, defined respectively as
\begin{equation}
\begin{aligned}
\Gamma_{\mathrm{TE}}(\varepsilon_{r},\theta)&=
\frac{\cos\theta-\sqrt{\varepsilon_{r}-\sin^{2}\theta}}{\cos\theta+\sqrt{\varepsilon_{r}-\sin^{2}\theta}},\\[2pt]
\Gamma_{\mathrm{TM}}(\varepsilon_{r},\theta)&=
\frac{\varepsilon_{r}\cos\theta-\sqrt{\varepsilon_{r}-\sin^{2}\theta}}{\varepsilon_{r}\cos\theta+\sqrt{\varepsilon_{r}-\sin^{2}\theta}}.
\end{aligned}
\label{eq:fresnel}
\end{equation}
Generally speaking, a path is a chain of $N_{\ell}$ specular reflections, where each path hits a surface of material $i$ defined by complex relative permittivity $\varepsilon^{(i)}$ at an incidence angle of $\theta^{(i)}$. Following \cite{10705152}, the specular field spreads as a spherical wave from the image source, hence, the path carries the free space path loss set by its \emph{total} length,
\begin{equation}
g_{\ell}=\frac{1}{d_{\ell}},\qquad d_{\ell}=\sum\nolimits_{i=1}^{N_{\ell}+1}d_{\ell}^{(i)},
\label{eq:pathloss}
\end{equation}
where $d_{\ell}^{(i)}$ is the length of the $i^{th}$ straight segment of path $\ell$, hence the $N_{\ell}+1$ segments run from the \ac{BS} to the first reflection point, between consecutive reflection points, and from the last reflection point to the \ac{UE}.
At each reflection $i$, the field is transformed by the diagonal Fresnel matrix $\boldsymbol{\Gamma}(\varepsilon^{(i)},\theta_{\ell}^{(i)})=\mathrm{diag}\big(\Gamma_{\mathrm{TE}},\Gamma_{\mathrm{TM}}\big)$, whose components are given in \eqref{eq:fresnel}. This transformation is applied in the local incidence frame of the surface, which is reached through the change of basis $\mathbf{W}_{\ell}^{(i)}$. The outgoing field therefore obeys
\begin{equation}
\mathbf{W}_{\ell}^{(i)}\,\mathbf{e}_{\ell,\mathrm{OUT}}^{(i)}
=\boldsymbol{\Gamma}\big(\varepsilon^{(i)},\theta_{\ell}^{(i)}\big)\,
\mathbf{W}_{\ell}^{(i)}\,\mathbf{e}_{\ell,\mathrm{IN}}^{(i)} .
\label{eq:refl}
\end{equation}
After the $N_{\ell}$ interactions the field arriving at the receiver is
$\mathbf{e}_{\ell,\mathrm{IN}}^{(N_{\ell}+1)}$, and the complex path gain is its
co-polarized projection onto the launched field $\mathbf{e}_{\ell,\mathrm{OUT}}^{(0)}$
\cite{10705152},
\begin{equation}
\alpha_{\ell}=g_{\ell}\,\frac{\big(\mathbf{e}_{\ell,\mathrm{OUT}}^{(0)}\big)^{ \mathsf T}
\,\mathbf{e}_{\ell,\mathrm{IN}}^{(N_{\ell}+1)}}
{\big\|\mathbf{e}_{\ell,\mathrm{OUT}}^{(0)}\big\|_{2}},
\qquad
\tau_{\ell}=\frac{1}{c}\sum\nolimits_{i=1}^{N_{\ell}+1}d_{\ell}^{(i)} .
\label{eq:gain}
\end{equation}
Equation \eqref{eq:gain} bundles the path loss \eqref{eq:pathloss} and the
reflections \eqref{eq:refl} into one scalar; for a single polarization with no cross
coupling it collapses to $\alpha_{\ell}=g_{\ell}\prod_{i}\Gamma(\varepsilon^{(i)},\theta_{\ell}^{(i)})=g_{\ell}R_{\ell}$,
with $g_{\ell}$ the free space path loss and $R_{\ell}$ the surface response,
recovering $\alpha_{\ell}=g_{\ell}\Gamma$ for a single reflection. In general, the
$\mathbf{W}_{\ell}^{(i)}$ contains the two polarizations.

When the \ac{DT} (or equivalently, the ray tracer it uses) assigns a wrong material, for example a glass material 
($\varepsilon_{\mathrm A}$) when the truth is concrete ($\varepsilon_{\mathrm B}$),
the geometry, delays $\tau_{\ell}$ and change of basis $\mathbf{W}_{\ell}^{(i)}$ are
unchanged but the reflection matrix $\boldsymbol{\Gamma}$ is impacted, so the overall error of a single path $\ell$ that experiences multiple bounces can be written as
\begin{equation}
\epsilon_{\ell}=\big[\alpha_{\ell}(\{\varepsilon_{\mathrm A}^{(i)}\})
-\alpha_{\ell}(\{\varepsilon_{\mathrm B}^{(i)}\})\big]\,e^{-j2\pi f\tau_{\ell}},
\label{eq:matpath}
\end{equation}
which can be expressed as $g_{\ell}[\Gamma(\varepsilon_{\mathrm A},\theta_{\ell})-\Gamma(\varepsilon_{\mathrm B},\theta_{\ell})]e^{-j2\pi f\tau_{\ell}}$ for the single reflection case.

As a consequence, the error a \ac{DT} commits is the deviation between the true model in \eqref{eq:multipath} from the assumed \ac{DT} model in \eqref{eq:twinpaths}, i.e. the \ac{DT} output differs from the true channel as
\begin{equation}
\begin{aligned}
\Delta \triangleq \hDT - h_{0} & = \sum\nolimits_{\ell} \epsilon_{\ell},
\end{aligned}
\label{eq:patherr}
\end{equation}
The randomness in \eqref{eq:patherr} is taken over a \emph{deployment
ensemble} $\mathcal{D}$, i.e. over all possible user positions, frequency and delays within the assumed environment, taken conditionally on the macroscopic channel $h_{0}$. For sake of compact representation, all expectations $\E_{\mathcal{D}}\{\cdot\}$ below are over $\mathcal{D}$, even when omitted. The errors are assumed to behave under the following assumption
\begin{assumption}[Lindeberg-Feller central limit theorem]\label{ass:clt}
Over the deployment ensemble $\mathcal{D}$ the channel errors committed by the \ac{DT}
$\{\epsilon_{\ell}\}_{\ell}$ in \eqref{eq:patherr} are independent across $\ell$,
circularly symmetric with finite variances
$\sigma_{\ell}^{2}\triangleq\mathrm{Var}_{\mathcal{D}}(\epsilon_{\ell})$. Writing
$s_{L}^{2}\triangleq\sum_{\ell}\sigma_{\ell}^{2}=\sDT$ for the total error power,
they satisfy \emph{Lindeberg's condition}: for every $\varepsilon>0$,
\begin{equation}
\frac{1}{s_{L}^{2}}\sum_{\ell}
\mathbb{E}_{\mathcal{D}} \Big[\,
\big|\epsilon_{\ell}-\mathbb{E}_{\mathcal{D}}\epsilon_{\ell}\big|^{2}\,
\mathbf{1}_{\{\,|\epsilon_{\ell}-\mathbb{E}_{\mathcal{D}}\epsilon_{\ell}|
>\varepsilon\, s_{L}\,\}}\Big]
\underset{L \to \infty}{\longrightarrow} 0.
\label{eq:lindeberg}
\end{equation}
\end{assumption}

Condition~\eqref{eq:lindeberg} relates to how the \ac{DT} modeling error is distributed across the paths over the deployment ensemble $\mathcal{D}$. It holds when the total error power
$\sDT$ is shared among many paths in the sense of \eqref{eq:lindeberg}, each carrying an error variance portion
$\sigma_\ell^2/s_L^2$, so that the \ac{DT} is comparably and mildly wrong on many paths at once. 
On the other hand, it fails when a single path holds most of the error variance, which then leads to heavy tails in the underlying distribution, hence the Gaussian model becomes no longer valid.
Under \textbf{Assumption~\ref{ass:clt}}, the Lindeberg-Feller central limit theorem applies to \eqref{eq:patherr}, where the normalized centered error
$(\Delta-b)/s_{L}$ converges in distribution, as $L\to\infty$, to a zero-mean
circularly symmetric standard Gaussian, in essence $\Delta \sim \CN \big(b,\,\sDT\big)$ where the bias term is $b \triangleq \E_{\mathcal{D}}\{\Delta\} = \sum_{\ell}\E_{\mathcal{D}}\{\epsilon_{\ell}\}$
and the variance is given as 
\begin{equation}
\sDT \triangleq \E_{\mathcal{D}}\{|\Delta-b|^{2}\}
     = \sum\nolimits_{\ell}\mathrm{Var}_{\mathcal{D}}(\epsilon_{\ell}).
\label{eq:s2dt}
\end{equation}
Writing the
centered fluctuation $e_{\mathrm{DT}}\triangleq\Delta-b\sim\CN(0,\sDT)$, the \ac{DT}
output is a \emph{noisy measurement} of the true channel,
\begin{equation}
\hDT = h_{0} + b + e_{\mathrm{DT}}, \qquad e_{\mathrm{DT}}\sim\CN(0,\sDT),
\label{eq:prior}
\end{equation}
which is a Gaussian observation of $h_{0}$ contributing Fisher information $1/\sDT$ as opposed to the
pilots that contribute $T_{\tau}\rho$. The variance $\sDT$ can be regarded as the amount of uncertainty the \ac{DT} holds, due to the mismatch between the \ac{DT} and the \ac{PW}. In short, the Gaussian model in \eqref{eq:prior} follows from the central limit theorem over the
many reflected paths under Assumption~\ref{ass:clt}.

We now confirm on a ray traced environment that a \ac{DT} with  wrong material can be captured by the two moments $(b,\sDT)$. Figure~\ref{fig:clt} depicts the histogram of the \ac{DT} error $\Delta $ in \eqref{eq:patherr} over the deployment ensemble $\mathcal{D}$ when the ray tracer assigns a wrong material to the surfaces, where its variance $\sDT$ depends on how wrong the assumed material is, and the bias $b$ appears when the path phases stay correlated across the deployment. In particular, we simulate the \ac{DT} measurement model \eqref{eq:prior} from the Fresnel coefficients in \eqref{eq:fresnel} and the full polarimetric gain \eqref{eq:gain}. We consider a $4\,\mathrm{m}\times4\,\mathrm{m}\times3\,\mathrm{m}$ room with metal walls, where the \ac{BS} is fixed and the \ac{UE} is swept over the room to average the entire deployment ensemble $\mathcal D$. The propagation paths are the specular reflection paths of the room, enumerated by the method of images, where each mirror image of the \ac{BS} across the walls is one reflected ray to the \ac{UE} traced through the reflection transformation \eqref{eq:refl} and projected by the gain \eqref{eq:gain}. 
The true surfaces are all metal (with $\varepsilon_{\mathrm{true}}=1-10^{7}j$), however the \ac{DT} assigns a wrong material, so that each reflection reflects with the incorrect Fresnel coefficients, hence  \eqref{eq:refl} is perturbed. In particular, in Figure \ref{fig:clt-a}, we simulate cases where the \ac{DT} assigns near metal, which is close to the true material ($\varepsilon_{\mathrm{DT}}=1-10^{3}j$), concrete ($\varepsilon_{\mathrm{DT}}=5.3-0.3j$), and a light dielectric material which is farthest from metal ($\varepsilon_{\mathrm{DT}}=1.6-0.01j$).
When the \ac{DT} assumes the wrong wall material, every reflected path picks up a small error in how strongly it reflects, and as these errors are independent across the many paths they add up constructively as powers, which contributes in widening the distribution of the total error. The further the assumed material deviates from the true metal (in terms of complex relative permittivity), the larger each path error becomes, so the variance $\sigma_{\mathrm{DT}}$ grows from $0.28$ for a near-metal assumption to $1$ for concrete to $1.07$ for a light dielectric, that is by a factor of about $3.8$ between the closest and the farthest material assumption.

In Figure~\ref{fig:clt-b} we fix the true material to metal (as is the case in Figure~\ref{fig:clt-a}), however the \ac{DT} assumes concrete material. The deployment ensemble here is considered to be a $2\,\mathrm{cm}$ patch. We study the effect of the carrier frequency on the \ac{DT} error $\Delta$, namely at a high carrier ($28\,\mathrm{GHz}$) the path phases decorrelate across the deployment ensemble and the same material error averages to a zero-mean spread with $b\approx0$, whereas at a low carrier ($3.5\,\mathrm{GHz}$) the phases stay correlated so the error adds coherently and a nonzero bias $b$ emerges.
When the wavelength is short relative to the region over which we average, e.g. at $28\,\mathrm{GHz}$, the phases of the reflected paths scramble from point to point, so the error cancels on average and the \ac{DT} error stays zero mean. On the other hand, when the wavelength is long compared to that region, e.g. at $3.5\,\mathrm{GHz}$, the phases slightly change relative to the wavelength, the same error adds up in the same direction, and it turns into a non-negligible bias.

\begin{figure*}[!t]
\centering
\resizebox{\textwidth}{!}{%
\begin{tikzpicture}[
  >=Stealth, font=\small,
  box/.style={rounded corners=2pt, draw=inkgray, semithick, align=center,
              inner sep=4pt, fill=white, minimum height=11mm},
  ph/.style ={rounded corners=2pt, draw=phys, semithick, align=center,
              inner sep=4pt, fill=phys!8, minimum height=11mm},
  dg/.style ={rounded corners=2pt, draw=dig, semithick, align=center,
              inner sep=4pt, fill=dig!12, minimum height=11mm},
  cmb/.style={rounded corners=2pt, draw=comb, thick, align=center,
              inner sep=5pt, fill=comb!8, minimum height=11mm},
  ar/.style ={-{Stealth[length=2.6mm]}, semithick, draw=inkgray},
  dar/.style={-{Stealth[length=2.6mm]}, semithick, draw=dig, dashed},
]
\begin{scope}[on background layer]
  \fill[phys!4,rounded corners=6pt] (-0.7,-0.95) rectangle (12.0,0.95);
  \fill[dig!6, rounded corners=6pt] (-0.7,-3.95) rectangle (12.0,-2.05);
\end{scope}
\node[anchor=west,text=phys,font=\bfseries] at (-0.65,1.25) {Physical world};
\node[anchor=west,text=dig, font=\bfseries] at (-0.65,-1.78) {Digital world};

\node[box,minimum width=11mm] (bs) at (0.4,0) {\textbf{BS}};
\node[box,minimum width=30mm] (env) at (3.7,0)
     {real propagation\\[-1pt]{\footnotesize geometry $+$ materials}};
\node[box,minimum width=11mm] (ue) at (7.0,0) {\textbf{UE}};
\node[ph,minimum width=36mm] (pw) at (10.4,0)
     {pilots $\bar y$\\[-1pt]{\footnotesize $\sim\CN \big(h_0,\,1/T_\tau\rho\big)$}};
\draw[ar] (bs) -- (env);
\draw[ar] (env) -- (ue);
\draw[ar] (ue) -- (pw);

\node[dg,minimum width=30mm] (rt) at (3.7,-3)
     {ray tracing replica\\[-1pt]{\footnotesize 3D model $+$ RT engine}};
\node[dg,minimum width=36mm] (dt) at (10.4,-3)
     {\ac{DT} $\hDT$\\[-1pt]{\footnotesize $\sim\CN \big(h_0+b,\,\sDT\big)$}};
\draw[ar] (rt) -- (dt);

\draw[dar] (env) -- node[right,align=left,font=\footnotesize]
     {renders\\(mismatch $b$,\,$\sDT$)} (rt);
\node[font=\footnotesize] at (8.2,-1.3) {$h_0=\sum_\ell \alpha_\ell e^{-j2\pi f\tau_\ell}$ (true channel)};

\node[cmb,minimum width=34mm] (cb) at (14.2,-1.5)
     {BLUE combiner\\[-1pt]{\footnotesize $\hPD=(1-w)\bar y+w\,\hDT$}};
\node[box,minimum width=20mm] (out) at (17.9,-1.5) {channel\\ estimate $\hPD$};
\draw[ar,rounded corners=2pt] (pw.east) -|
     node[near start,above,font=\footnotesize]{$1-w$} (cb.north);
\draw[ar,rounded corners=2pt] (dt.east) -|
     node[near start,below,font=\footnotesize]{$w$} (cb.south);
\draw[ar] (cb) -- (out);
\draw[ar] (out) -- ++(0,-1.05)
     node[below,align=center,font=\footnotesize]{data transmission\\ \& rate};
\end{tikzpicture}}
\caption{The receiver estimates the channel $h_0$ from two measurements: a \emph{physical} one (pilots $\bar y$, giving $\hPW$) and a \emph{digital} one (the ray tracing \ac{DT} $\hDT$), fused by the \ac{BLUE} estimator.}
\label{fig:sysmodel}
\end{figure*}
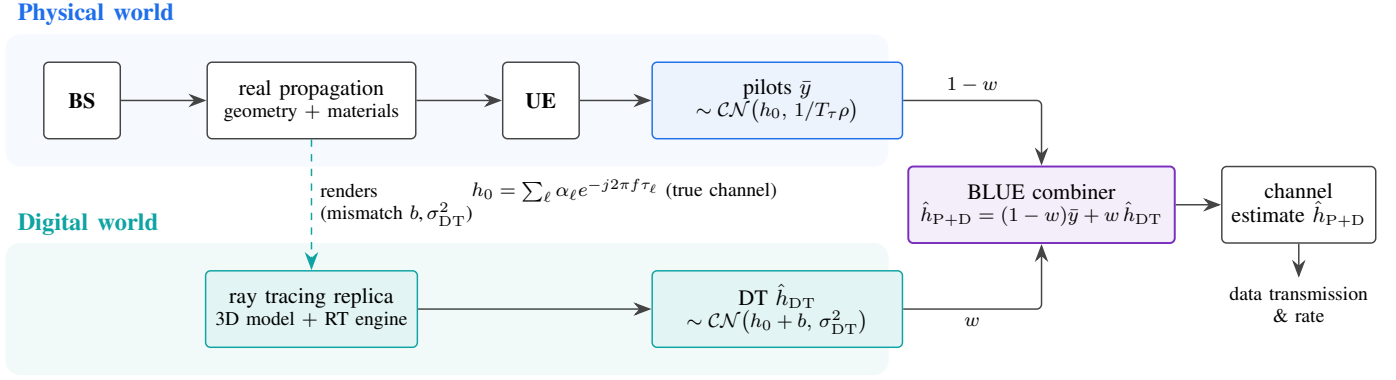

\section{Channel Estimation via \ac{DT}-aided Wireless System}\label{sec:analysis}
We consider three estimators of the deterministic $h_{0}$, distinguished by the information each uses.

\begin{itemize}
\setlength{\itemsep}{4pt}
\item[\emph{(i)}] \emph{Physical world}: We refer to the channel estimation procedure using pilots-only as the \emph{physical world} as it entails conducting a channel estimate solely using the pilots on a physical channel. It is well known that the \ac{ML} estimator from the pilots alone using the model in \eqref{eq:pilot} is given as follows
\begin{equation}
\begin{aligned}
\hPW = \argmax_{h}\, p\big(\{y_{i}\}\mid h_0\big) = \bar{y}.
\end{aligned}
\label{eq:pilotmse}
\end{equation}
Here $p(\{y_i\} \mid h)$ is the likelihood of the channel $h_0$ given the pilot observations $\{y_i\}$, that is, the joint density of the pilots evaluated as a function of $h_0$. Under the Gaussian pilot model in \eqref{eq:pilot}, the maximizer is the sample mean $\bar{y}$, so $\hPW$ is the \ac{MLE} formed from the physical world measurements alone. The \ac{MLE} estimator for $h_{0}$ is nothing other than the mean
\begin{equation}
\bar{y} \triangleq \frac{1}{T_{\tau}}\sum\nolimits_{i=1}^{T_{\tau}} y_{i}
\sim \CN \left(h_{0}, \frac{1}{T_{\tau}\rho}\right),
\label{eq:ybar}
\end{equation}
The \ac{MSE} of the estimation procedure can be shown to be \cite{kay2009fundamentals}
\begin{equation}
\label{eq:EPW}
	\mathcal{E}_{\mathrm{PW}} \triangleq \E\{|\hPW-h_{0}|^{2}\}=\frac{1}{T_{\tau}\rho},
\end{equation}
so that $T_{\tau}$ pilots contribute Fisher information $T_{\tau}\rho$ about $h_{0}$, which indicates that the \ac{MSE} on the channel estimation in the \ac{PW} improves only when more overhead is required on training the channel via increasing $T_\tau$ or when the system operates at a higher \ac{SNR} $\rho$.
\item[\emph{(ii)}] \emph{Digital world}: Given a \ac{DT} whose  measurement model $\hDT$ follows equation \eqref{eq:twinpaths} and whose errors follow \textbf{Assumption \ref{ass:clt}}, one can assert the equivalent Gaussian model in equation \eqref{eq:prior}. It immediately follows that the error on the \ac{DT} measurement is given as 
\begin{equation}
\mathcal{E}_{\mathrm{DT}} \triangleq \E\{|\hDT-h_{0}|^{2}\}=|b|^{2}+\sDT.
\label{eq:dwmse}
\end{equation}
Equation \eqref{eq:dwmse} tells us that the error depends naturally on the bias and the error variance of the \ac{DT}. 
\item[\emph{(iii)}] \emph{Physical and digital worlds}: The first two optimal estimators could be seen as two extremes, namely, the \ac{PW} channel estimator in \eqref{eq:pilotmse} is a typical channel estimator based on \ac{MLE} when no additional information about the channel is available. In other words, it relies on physical world observations only utilizing a known pilot sequence for channel estimation purposes. On the other hand, the second estimator using a \ac{DT} measurement following \eqref{eq:prior} relies on prior knowledge of the channel and accuracy of the \ac{DT} generator captured through $b$ and $\sDT$. In order to study the efficacy of having both \ac{PW} and \ac{DT} measurements, we can treat the \ac{DT} output \eqref{eq:prior} as a second independent measurement of $h_{0}$ that can aid the \ac{PW} channel estimation procedure. In essence, under the unbiased \ac{DT} case (i.e. $b=0$), the following weighted least squares problem can be considered
\begin{equation}
\hPD=\argmin_{h} T_{\tau}\rho\,\big|\bar{y}-h\big|^{2}
        +\tfrac{1}{\sDT}\,\big|\,\hDT-h\,\big|^{2},
\label{eq:mapls}
\end{equation}
whose minimizer is the inverse variance combination as 
\begin{equation}
\hPD = (1-w)\,\bar{y} + w\,\hDT, \qquad
w \triangleq \frac{1/\sDT}{\,T_{\tau}\rho + 1/\sDT\,}.
\label{eq:shrink}
\end{equation}
We consider the above estimator because under the Gauss-Markov theorem, the estimator in \eqref{eq:shrink} is also the \ac{BLUE} estimator, which achieves the lowest sampling variance among all linear unbiased estimators. Here, $w\in[0,1]$ is the weight on the \ac{DT}, and $\hPD$ interpolates between the other two estimators, namely $\hPD=\hPW$ at $w=0$, which is the case of an uninformative \ac{DT} and $\hPD=\hDT$ at $w=1$, which is the case of a perfectly confident \ac{DT}.
\end{itemize}

For a given \ac{DT}-aided system, one may be interested in the smallest channel estimation error achievable, i.e. how much can the \ac{DT} plus pilots enhance the \ac{MSE}. For that, we have the following proposition
\begin{proposition}[\ac{DT}-aided CRB]\label{prop:crb}
For an unbiased \ac{DT}, the minimum variance of any
unbiased estimator of $h_{0}$, attained by the \ac{BLUE} estimator in \eqref{eq:shrink} attains the \ac{CRB} of the joint model observing \eqref{eq:pilot} and \eqref{eq:prior} as
\begin{equation}
\mathcal{E}_{\mathrm{P+D}}(T_{\tau}) \triangleq \E\big\{|\hPD-h_{0}|^{2}\big\}
= \frac{1}{\,T_{\tau}\rho + 1/\sDT\,},
\label{eq:crbdt}
\end{equation}
with $\hPD$ by the \ac{BLUE} estimator in \eqref{eq:shrink} and the expectation taken over the
thermal noise $n_{i}$ and the \ac{DT} error variance $e_{\mathrm{DT}}$.
\end{proposition}
\begin{IEEEproof}
For an unbiased \ac{DT}, the pilots and the \ac{DT} are two \emph{independent} unbiased
Gaussian measurements of the same deterministic scalar $h_{0}$: the pilot average
$\bar{y}\sim\CN \big(h_{0},1/(T_{\tau}\rho)\big)$ from \eqref{eq:ybar}, carrying
Fisher information $T_{\tau}\rho$, and the \ac{DT} $\hDT\sim\CN(h_{0},\sDT)$ from
\eqref{eq:prior} with $b=0$, carrying Fisher information $1/\sDT$. For
independent observations Fisher information adds, so the information about
$h_{0}$ is $J=T_{\tau}\rho+1/\sDT$ and the \ac{CRB} is $1/J$. The \ac{CRB} is met by the BLUE in \eqref{eq:shrink} with error $J^{-1} = (T_{\tau}\rho+1/\sDT)^{-1}$, which is \eqref{eq:crbdt}. An alternative proof follows by using the expression of the \ac{BLUE} estimator in \eqref{eq:shrink}, where one can replace $\bar{y}$ with $h_0 + \bar{n}$ where $\bar{n} \sim \mathcal{CN}(0,\frac{1}{T_\tau \rho})$ and $\hDT$ with $h_0  + e_{\mathrm{DT}}$.
\end{IEEEproof}
In the context of \acp{DT} with pilots, equation \eqref{eq:crbdt}, hereby referred to as \emph{\ac{DT}-aided} bound, reveals two extremes, i.e.
\begin{equation}
\begin{aligned}
\lim_{\sDT\to\infty} \mathcal{E}_{\mathrm{P+D}} &= \frac{1}{T_{\tau}\rho} = \mathcal{E}_{\mathrm{PW}}\ \ (\text{pilot-only}),\\
\lim_{\sDT\to 0} \mathcal{E}_{\mathrm{P+D}} &= 0\ \ (\text{genie \ac{DT}}).
\end{aligned}
\label{eq:limits}
\end{equation}
The first limit tells us that, for the case of unbiased \ac{DT}, when the \ac{DT} becomes highly uninformative, one would have to rely on the physical world through pilot training only, and hence the performance is dictated by either the \ac{SNR} or training time. On the other hand, and for the case of unbiased \ac{DT}, the more informative the \ac{DT} becomes, the less one can rely on the \ac{PW}, i.e. one would rely less on the \ac{SNR} or the training time. The extreme case is a genie \ac{DT} which is capable of generating the true channel, on average.
Next, we target a question of how much training in the \ac{PW} a \ac{DT}-aided system requires to reach a desired \ac{MSE} on the channel estimate. 
\begin{corollary}[Pilots for a target accuracy]\label{cor:pilots}
The training length required to reach a target \ac{MSE} $\epsilon$ of a \ac{DT}-aided system observing measurements from \eqref{eq:pilot} and \eqref{eq:prior} is given as
\begin{equation}
T_{\tau}^{\star}(\epsilon,\sDT) = \frac{1}{\rho}\left(\frac{1}{\epsilon}-\frac{1}{\sDT}\right)^{ +}.
\label{eq:npilots}
\end{equation}
\end{corollary}
\begin{IEEEproof}
By \eqref{eq:crbdt}, $\mathcal{E}_{\mathrm{P+D}}(T_{\tau})$ is strictly decreasing
in $T_{\tau}\ge0$. Setting \eqref{eq:crbdt} equal to the target \ac{MSE}, ones gets $\mathcal{E}_{\mathrm{P+D}}(T_{\tau})=\epsilon$,
which in turn reads $T_{\tau}\rho+1/\sDT=1/\epsilon$, hence
$T_{\tau}=(1/\epsilon-1/\sDT)/\rho$. As $T_{\tau} \geq 0$, we get \eqref{eq:npilots}. 
\end{IEEEproof}
Indeed, a direct consequence of \textbf{Corollary \ref{cor:pilots}} is that when the \ac{DT} alone already meets the
target \ac{MSE} $\epsilon$, i.e.\ $\sDT\le\epsilon$, the
right hand side becomes zero, which indicates that no training is needed. In retrospect, if the \ac{DT} does not contain enough information, arising from modeling inaccuracies, i.e. $1/\epsilon >  1/\sDT$, then one would require $T_{\tau}^{\star}(\epsilon,\sDT)$ to reach the desired \ac{MSE}, i.e. $\epsilon$. From \textbf{Corollary \ref{cor:pilots}}, one can get a sense of how much training could be \textit{bypassed}, as a function of the accuracy of the \ac{DT} in question. Next, one may wonder how much pilots does a \ac{DT} replace. Said differently, how many pilots in a \ac{PW} does a \ac{DT} replace. For that, we introduce the following pilot equivalence law for an unbiased \ac{DT},

\begin{lemma}[Pilot-equivalence law of an unbiased \ac{DT}]\label{lem:equiv}
An unbiased \ac{DT} of variance $\sDT$ is worth
\begin{equation}
T_{\mathrm{eq}} = \frac{1}{\rho\,\sDT}
\label{eq:neq}
\end{equation}
training symbols, in the sense that
$\mathcal{E}_{\mathrm{P+D}}(T_{\tau})=1/\big((T_{\tau}+T_{\mathrm{eq}})\rho\big)$, i.e.\
estimating with the \ac{DT} and $T_{\tau}$ training symbols equals estimating with
$T_{\tau}+T_{\mathrm{eq}}$ training symbols and no \ac{DT} at our disposal.
\end{lemma}
\begin{IEEEproof}
By the \ac{DT}-aided bound in \eqref{eq:crbdt} at $T_\tau$ pilots, we simply equate the two errors in equations \eqref{eq:crbdt} and \eqref{eq:EPW}
\begin{equation*}
\mathcal{E}_{\mathrm{P+D}}(T_{\tau})=\frac{1}{(T_{\tau}+T_{\mathrm{eq}})\rho}
=\mathcal{E}_{\mathrm{PW}}(T_{\tau}+T_{\mathrm{eq}}),
\end{equation*}
With some straightforward manipulations, the proof is done.
\end{IEEEproof}
In other words, if the \ac{SNR} is high enough, then the \ac{DT} may turn out to be less useful in a sense it replaces less pilots. That is to say that the number of pilots a \ac{DT} replaces depends not only on the information contained in the \ac{DT}, i.e. $\frac{1}{\sDT}$, but also on the operating \ac{SNR} in the \ac{PW}.

We now let the \ac{DT} be biased, i.e. $\E\{\hDT\}=h_{0}+b$ with $b\neq 0$ in \eqref{eq:prior}, while the \ac{BLUE} in 
\eqref{eq:shrink} still trusts the \ac{DT} with weight $w$, i.e. the \ac{BLUE} estimator still thinks the \ac{DT} is unbiased. Therefore, we analyze the error in terms of channel estimation on using a biased \ac{DT}.

\begin{proposition}[Biased \ac{DT}]\label{prop:mcrb}
The \ac{MSE} of the \ac{DT}-aided estimator given in \eqref{eq:shrink} under an unknown $b$ is the unbiased bound plus a bias penalty,
\begin{equation}
\begin{aligned}
\mathcal{E}_{\mathrm{P+D}}^b(T_{\tau},b)
&\triangleq \E\big\{|\hPD-h_{0}|^{2}\big\}\\
&= \mathcal{E}_{\mathrm{P+D}}(T_{\tau})
+ w^{2}\,|b|^{2}.
\end{aligned}
\label{eq:mcrb}
\end{equation}
\end{proposition}
\begin{IEEEproof}
Notice that the \ac{BLUE} estimator in \eqref{eq:shrink}
can be expressed as 
\begin{equation*}
	\begin{split}
		\vert \hPD - h_0 \vert ^2 &=\vert(1-w)\bar{y}+w\,\hDT - h_0 \vert^2 \\
		&=\vert (1-w)(h_0 + \bar{n}) +w(h_0 +b + e_{\mathrm{DT}}) - h_0 \vert^2 \\
		&=\vert w\,b+(1-w)\bar{n}+w\,e_{\mathrm{DT}} \vert^2,
	\end{split}
\end{equation*}
where $\bar{n} \triangleq \frac{1}{T_{\tau}}\sum_{i=1}^{T_{\tau}} n_{i} \sim \mathcal{CN}(0, \sigma_n^2/T_\tau)$,
then applying the expectation, we get
\begin{equation*}
\begin{split}
\mathbb{E} \vert \hPD - h_0 \vert ^2
&= w^{2}|b|^{2} + (1-w)^{2}\,\mathbb{E}|\bar{n}|^{2} + w^{2}\,\mathbb{E}|e_{\mathrm{DT}}|^{2} \\
&= w^{2}|b|^{2} + (1-w)^{2}\,\frac{1}{T_{\tau}\rho} + w^{2}\sigma_{\mathrm{DT}}^{2}.
\end{split}
\label{eq:mse-expand}
\end{equation*}
Using $w = \frac{1/\sDT}{\,T_{\tau}\rho + 1/\sDT\,}$ in the above expression gives us \eqref{eq:mcrb}.
\end{IEEEproof}
The biased \ac{DT} reveals several insights. 
First, an uninformative \ac{DT} ($\sDT\to\infty$) behaves as an asymptotically unbiased \ac{DT} in a sense that $\mathcal{E}_{\mathrm{P+D}}^b(T_{\tau},b) \to \mathcal{E}_{\mathrm{P+D}}(T_{\tau})$ as $\sDT\to\infty$. In other words, the bias can be overshadowed by high error variance of the \ac{DT}.
Second, a biased \ac{DT} can drive the \ac{MSE} error  $\mathcal{E}_{\mathrm{P+D}}^b(T_{\tau},b)$ high enough if more weight $w$ is placed on trusting the \ac{DT}. The next proposition gives us the exact bias threshold that, if trespassed, would yield a destructive \ac{DT} in the \ac{MSE} sense.

\begin{proposition}[Negative value of a \ac{DT}]\label{prop:neg}
For the case of a biased \ac{DT} with unknown bias $b$, we have that
$\mathcal{E}_{\mathrm{P+D}}^b>\mathcal{E}_{\mathrm{PW}}$, if and only if the bias exceeds
\begin{equation}
|b^{\star}| = \sqrt{\frac{1}{T_{\tau}\rho}+\sDT} .
\label{eq:bstar}
\end{equation}
\end{proposition}
\begin{IEEEproof}
By \eqref{eq:mcrb} and \eqref{eq:EPW},
$\mathcal{E}_{\mathrm{P+D}}>\mathcal{E}_{\mathrm{PW}}$ reads
$w^{2}|b|^{2}>\dfrac{1}{T_{\tau}\rho}-\dfrac{1}{T_{\tau}\rho+1/\sDT}
=\dfrac{1/\sDT}{T_{\tau}\rho\,(T_{\tau}\rho+1/\sDT)}$. Using $w$ in \eqref{eq:shrink} rearranges to
$|b|^{2}>\dfrac{T_{\tau}\rho+1/\sDT}{T_{\tau}\rho\,(1/\sDT)}=\dfrac{1}{T_{\tau}\rho}+\sDT = |b^{\star}|^2$,
i.e.\ \eqref{eq:bstar}.
\end{IEEEproof}
It is worth noting that \textbf{Proposition~\ref{prop:neg}} in other words gives us the bias threshold $b^{\star}$ in which \emph{trusting a biased \ac{DT} is strictly worse than ignoring it}, which happens when the unknown bias exceed $b^{\star}$.
In particular, the proposition describes three regimes: for $|b|<|b^{\star}|$
the \ac{DT} \emph{helps}; as $\sDT$ grows the gain \emph{saturates} at the
pilot only error \eqref{eq:pilotmse}; and for $|b|>|b^{\star}|$ the \ac{DT}
\emph{hurts}. 
By \eqref{eq:bstar}, the bias tolerance \emph{shrinks} with more
pilots or high \ac{SNR}, namely $|b^{\star}|^{2}=1/(T_{\tau}\rho)+\sDT \to \sDT$.
On the other hand, with few pilots, or even low \ac{SNR}, a
strongly biased \ac{DT} can still help with channel estimation tasks. However, the above proposition only quantifies the biased \ac{DT} from a channel estimation perspective and does not tell us anything about the required number of pilots needed to attain a desired \ac{MSE}. Note that \textbf{Corollary \ref{cor:pilots}} answers that question for the case of an unbiased \ac{DT}. For the case of a biased \ac{DT} with unknown bias, we have the following generalized corollary

\begin{corollary}[Training for a target with a biased \ac{DT}]\label{cor:pilotsmis}
For a biased \ac{DT} of unknown bias $b$ and variance $\sDT$, the training length to reach a
attain an \ac{MSE} requirement $\epsilon$ is given as 
\begin{equation}
T_{\tau}^{\star}(\epsilon,b)=\frac{1}{\rho}\left[
\frac{1+\sqrt{\,1+4\epsilon\,|b|^{2}/\sigma_{\mathrm{DT}}^{4}\,}}{2\epsilon}
-\frac{1}{\sDT}\right]^{+}.
\label{eq:npilotsmis}
\end{equation}
The pilot length above recovers \textbf{Corollary~\ref{cor:pilots}} at $b=0$ and requires no training,
$T_{\tau}^{\star}=0$, exactly when $\sDT+|b|^{2}\le\epsilon$, i.e.\ the \ac{DT} alone
already meets the target.
\end{corollary}
\begin{IEEEproof}
First, we set $\mathcal{E}_{\mathrm{P+D}}^b(T_{\tau},b)=\epsilon$. Using \eqref{eq:mcrb}, we have that
\begin{equation}
\frac{1}{T_{\tau}\rho+1/\sDT}+w^{2}|b|^{2}=\epsilon,
\label{eq:biased-condition}
\end{equation}
where $
w=\frac{1/\sDT}{T_{\tau}\rho+1/\sDT}$. Denoting $x\triangleq T_{\tau}\rho+1/\sDT>0$ gives $w=(1/\sDT)/x$ and
$w^{2}|b|^{2}=|b|^{2}/(\sigma_{\mathrm{DT}}^{4}x^{2})$, so that the condition in equation \eqref{eq:biased-condition} becomes
$\tfrac{1}{x}+\tfrac{|b|^{2}}{\sigma_{\mathrm{DT}}^{4}x^{2}}=\epsilon$, i.e. quadratic in $x$, namely
\begin{equation*}
\epsilon\,x^{2}-x-\frac{|b|^{2}}{\sigma_{\mathrm{DT}}^{4}}=0 .
\end{equation*}
The unique positive root is easily verified to be
$x^+=\big(1+\sqrt{1+4\epsilon|b|^{2}/\sigma_{\mathrm{DT}}^{4}}\big)/(2\epsilon)$.
Since $T_{\tau}=(x-1/\sDT)/\rho$ and the training length is nonnegative, then replacing $x^+$ in the expression of $T_\tau$ gives us \eqref{eq:npilotsmis}. 
We apply the $[\cdot]^+$ operator so that the training length is non-negative. 
Moreover, it is easy to see that at $\sDT + |b|^{2} = \epsilon$, the expression in \eqref{eq:npilotsmis} is exactly zero. Furthermore, since \eqref{eq:npilotsmis} is decreasing in $\epsilon$ (for fixed $b$ and $\sDT$), then any $\sDT + |b|^{2} \leq \epsilon$ drives the bracketed term to negative, which then attains $T_{\tau}^{\star}(\epsilon,b)=0$.  
\end{IEEEproof}

More specifically, \textbf{Corollary \ref{cor:pilotsmis}} indicates that when the biased \ac{DT} alone already meets the target \ac{MSE} $\epsilon$, i.e.\ $\sDT + |b|^{2} \leq \epsilon$, the right hand side becomes zero, hence zero training is required. It is important to see that the training required as a function of \ac{MSE} requirement $\epsilon$ is dictated not only by the variance $\sDT$ but also by the unknown bias $b$. Also, notice how the bias and variance are treated symmetrically only at the boundary required for \ac{PW} training.

Another insight worth highlighting is the behavior of \eqref{eq:npilotsmis} with $\sDT$. Indeed, differentiating \eqref{eq:npilotsmis} with respect to $\sDT$ gives
\begin{equation}
\begin{split}
	\rho\,\frac{\partial T_{\tau}^{\star}}{\partial \sigma_{\mathrm{DT}}^{2}}
=\frac{1}{\sigma_{\mathrm{DT}}^{4}}-\frac{2|b|^{2}/\sigma_{\mathrm{DT}}^{6}}{\sqrt{1+4\epsilon|b|^{2}/\sigma_{\mathrm{DT}}^{4}}}.
\end{split}
\end{equation}
and setting it to zero gives us
\begin{equation}
\frac{1}{\sigma_{\mathrm{DT}}^{4}}=\frac{2|b|^{2}/\sigma_{\mathrm{DT}}^{6}}{\sqrt{1+4\epsilon|b|^{2}/\sigma_{\mathrm{DT}}^{4}}}
\quad\Longrightarrow\quad
\sqrt{1+\frac{4\epsilon|b|^{2}}{\sigma_{\mathrm{DT}}^{4}}}=\frac{2|b|^{2}}{\sigma_{\mathrm{DT}}^{2}},
\end{equation}
so the optimal \ac{DT} uncertainty is
\begin{equation}
\sigma_{\mathrm{DT}}^{2\star}=2|b|\sqrt{|b|^{2}-\epsilon},\qquad |b|^{2}>\epsilon,
\label{eq:opt-sdt}
\end{equation}
and for $|b|^{2}\leq \epsilon$, the radicand is negative and no interior optimum exists, which means that the minimum is obtained at  $\sigma_{\mathrm{DT}}^{2\star}=0$. Now, we can introduce the pilot equivalence law for the case of a biased \ac{DT}

\begin{corollary}[Pilot-equivalence law of a biased \ac{DT}]\label{cor:equivmis}
A biased \ac{DT} with unknown bias $b$ and variance $\sDT$ is worth
\begin{equation}
T_{\mathrm{eq}}^{b}(T_{\tau},b)=T_{\mathrm{eq}}\cdot
\frac{1+T_{\tau}\rho\,(\sDT-|b|^{2})}{1+T_{\tau}\rho\,\sDT+|b|^{2}/\sDT}
\label{eq:neqmis}
\end{equation}
training symbols when the \ac{PW} is trained on $T_\tau$ training symbols, in the sense that
$\mathcal{E}_{\mathrm{P+D}}(T_{\tau},b)=1/\big((T_{\tau}+T_{\mathrm{eq}}^{b})\rho\big)$.
The expression in \eqref{eq:neqmis} coincides with \textbf{Lemma~\ref{lem:equiv}} at $b=0$,
vanishes at $|b|^{2}=\sDT+1/(T_{\tau}\rho)=|b^{\star}|^{2}$ of \eqref{eq:bstar},
and is \emph{negative} beyond it: a strongly biased twin is worth a negative
number of pilots.
\end{corollary}
\begin{IEEEproof}
By definition, and according to the above pilot-equivalence law, $T_{\mathrm{eq}}^{b}$ satisfies
$\mathcal{E}_{\mathrm{PW}}(T_{\tau}+T_{\mathrm{eq}}^{b})=\mathcal{E}_{\mathrm{P+D}}(T_{\tau},b)$,
hence by \eqref{eq:pilotmse}
$T_{\mathrm{eq}}^{b}=1/\big(\rho\,\mathcal{E}_{\mathrm{P+D}}(T_{\tau},b)\big)-T_{\tau}$.
Substituting $\mathcal{E}_{\mathrm{P+D}}(T_{\tau},b)$ from \eqref{eq:mcrb} and
clearing denominators, the numerator collapses to
$1+T_{\tau}\rho(\sDT-|b|^{2})$ and the denominator to
$\rho\sDT\big(1+T_{\tau}\rho\sDT+|b|^{2}/\sDT\big)$, which is \eqref{eq:neqmis}.
The numerator vanishes exactly at $|b|^{2}=\sDT+1/(T_{\tau}\rho)$, i.e.\
\eqref{eq:bstar}, and is negative for larger $|b|$.
\end{IEEEproof}

Unlike \textbf{Lemma \ref{lem:equiv}}, the biased \ac{DT} contains a factor, which multiplies $T_{\mathrm{eq}}$ that depends on the bias itself, as well as the amount of training used in the \ac{PW} via $T_\tau$.

However, an important regime to note is the high \ac{SNR} regime, where the biased \ac{DT} is asymptotically worth 
\begin{equation}
T_{\mathrm{eq}}^{b}(T_{\tau},b) \propto T_{\mathrm{eq}}\cdot
\left( 1- \frac{|b|^{2}}{\sDT} \right),
\end{equation}
which indicates that even when operating with high \ac{SNR} in the \ac{PW}, the bias can penalize part of the pilot equivalent number of pilots. 
Generally speaking, \textbf{Corollary \ref{cor:equivmis}} tells us that when $\vert b \vert^2 <  |b^{\star}|^{2}$, the equivalent amount of pilots the \ac{DT} is worth (i.e. the \ac{DT} replaces), namely $T_{\mathrm{eq}}^{b}(T_{\tau},b)$, is positive. At the limit, $\vert b \vert^2 = |b^{\star}|^{2}$, we have that $T_{\mathrm{eq}}^{b}(T_{\tau},b^{\star})=0$ indicating that the \ac{DT} is worth zero pilots. On the other hand, once the bias exceeds $\vert b \vert^2 > |b^{\star}|^{2}$, the equivalent amount of pilots the \ac{DT} replaces becomes negative, i.e. $T_{\mathrm{eq}}^{b}(T_{\tau},b)<0$ indicating that the \ac{DT} starts removing the amount of actual pilots, i.e. $T_\tau + T_{\mathrm{eq}}^{b}(T_{\tau},b) < T_\tau$, namely an \emph{unlearning phenomenon of the \ac{PW}}.

\section{Capacity Analysis of a \ac{DT}-aided Wireless System}\label{sec:cap}
We now translate the estimation results into achievable rate via the training
based capacity lower bound of \cite{hassibi2003howmuch}. Herein, we adopt the standard
block fading model, namely the unit variance channel $h_{0}\sim\CN(0,1)$ is constant
over a coherence interval of $T$ symbols, of which $T_{\tau}$ carry training and
$T-T_{\tau}$ carry data, at common SNR $\rho$. The receiver uses the pilots
$\bar y$ and the \ac{DT} output $\hDT$ into a single channel estimate and decodes
the data with it.
We now introduce analyze the achievable rate a \ac{DT}-aided wireless system delivers. For that, we have the following result
\begin{lemma}[\ac{DT}-aided achievable rate]\label{lem:cap}
For the block fading model above with a unbiased \ac{DT} of precision $1/\sDT$, the
achievable rate is lower bounded by

\begin{equation}
C(T_{\tau};\rho)=\frac{T-T_{\tau}}{T}\,\log_{2} \frac{\left(1+\tfrac{1}{\sDT}+\rho T_{\tau}\right)\left(1+\rho\right)}{1+\tfrac{1}{\sDT}+\rho T_{\tau}+\rho}.
\label{eq:capcf}
\end{equation}
For the biased case, we have 
\begin{equation}
C^{b}(T_{\tau};\rho)=\frac{T-T_{\tau}}{T}\,\log_{2}\frac{1+\rho}{1+\rho\big(\frac{1}{1+\rho T_{\tau}+1/\sDT}+w_b^{2}|b|^{2}\big)},
\label{eq:capcfb}
\end{equation}
where $w_b \triangleq \frac{1/\sDT}{1+\rho T_{\tau}+1/\sDT}$.
\end{lemma}
\begin{IEEEproof}
We assume that the data phase is conducted after the system acquires a channel estimate using both the \ac{PW} observation of $h_{0}$ via the pilot average $\bar y=h_{0}+n$,
$n\sim\CN(0,(\rho T_{\tau})^{-1})$ from \eqref{eq:ybar}, and the unbiased \ac{DT} from \eqref{eq:prior} and that $n$, $e$ and the
prior $h_{0}\sim\CN(0,1)$ are mutually independent.
Furthermore, we stack them as
$\mathbf{r}=[\,\bar y,\ \hDT\,]^{\top}=\mathbf{1}\,h_{0}+\mathbf{v}$, with
$\mathbf{1}=[1,1]^{\top}$ and $\mathbf{v}\sim\CN(\mathbf{0},\mathbf{R})$,
$\mathbf{R}=\mathrm{diag} \big((\rho T_{\tau})^{-1},\,\sDT\big)$, the \ac{MMSE}
estimate of $h_{0}$ is the conditional mean
\begin{equation}
\hat h=\E\{h_{0}\mid\bar y,\hDT\}
      =\frac{\rho T_{\tau}\,\bar y+(1/\sDT)\,\hDT}{1+\rho T_{\tau}+1/\sDT},
\label{eq:cmb}
\end{equation}

Since $h_{0}$ and $\mathbf{v}$ are jointly Gaussian, the posterior $p(h_{0}\mid\mathbf{r})$ is Gaussian, and gathering the quadratic in $h_{0}$ from the product of the likelihood $\exp[-(\mathbf{r}-\mathbf{1}h_{0})^{H}\mathbf{R}^{-1}(\mathbf{r}-\mathbf{1}h_{0})]$ and the prior $\exp[-|h_{0}|^{2}]$ yields
\begin{equation}
p(h_{0}\mid\mathbf{r})\ \propto\ \exp \big[-u\,|h_{0}-\hat h|^{2}\big],
\label{eq:posterior}
\end{equation}
so the posterior is $\mathcal{CN}(\hat h,\,(1+\mathbf{1}^{H}\mathbf{R}^{-1}\mathbf{1})^{-1})$ with $\hat h$ the conditional mean of \eqref{eq:cmb}. The inverse of the variance $1+\mathbf{1}^{H}\mathbf{R}^{-1}\mathbf{1}$ is the so-called posterior precision in information form, namely the prior precision added to the measurement precision $\mathbf{1}^{H}\mathbf{R}^{-1}\mathbf{1}$ that the two observations carry about $h_{0}$. The \ac{MMSE} equals this posterior variance, and it is independent of the realization $\mathbf{r}$, so the error $\tilde h=h_{0}-\hat h$ has variance
\begin{equation}
\sigma_{e}^{2}=\big(1+\mathbf{1}^{H}\mathbf{R}^{-1}\mathbf{1}\big)^{-1}
             =\frac{1}{1+\rho T_{\tau}+1/\sDT},
\label{eq:se2}
\end{equation}
and, by the orthogonality principle of the \ac{MMSE} between its estimate and error, namely $\hat h\perp\tilde h$ with $\E|\hat h|^{2}=1-\sigma_{e}^{2}$.
Conditioning the data phase on $\hat h$ (a sufficient statistic for $h_{0}$) is
thus conditioning on both $\bar y$ and $\hDT$.
Furthermore, with a unit power data symbol $x$, the receiver \ac{UE} reads
\begin{equation*}
y=\sqrt{\rho}\,h_{0}\,x+w
 =\sqrt{\rho}\,\hat h\,x+\big(\sqrt{\rho}\,\tilde h\,x+w\big),
\end{equation*}
where the noise behaves as
$w\sim\CN(0,1)$. Denote $z \triangleq \sqrt{\rho}\,\tilde h\,x+w$, then 
the effective noise $z$ is uncorrelated from $x$ because
$\E\{z\,x^{*}\}=\sqrt{\rho}\,\E\{\tilde h\}\,\E|x|^{2}+\E\{w\,x^{*}\}=0$ ($\tilde h$
is zero mean and independent of $x$, and $w$ is independent), with power
$\E|z|^{2}=1+\rho\sigma_{e}^{2}$, while the useful signal power is
$\rho\,\E|\hat h|^{2}=\rho(1-\sigma_{e}^{2})$.
Conditioned on $\hat h$, the input sees a known
gain and additive noise uncorrelated with it at fixed variance; by the worst
case uncorrelated noise theorem \cite[Thm.~1]{hassibi2003howmuch} a Gaussian $z$
minimizes the mutual information, hence
\begin{equation}
I(x;y\mid\hat h) \ge \log_{2} \big(1+\rho_{\mathrm{eff}}\big),
\label{eq:hhcap}
\end{equation}
where
\begin{equation}
	\rho_{\mathrm{eff}}=\frac{\rho\,(1-\sigma_{e}^{2})}{1+\rho\,\sigma_{e}^{2}}.
	\label{eq:rho_eff}
\end{equation} 
Note that only $T-T_{\tau}$ of every $T$ symbols
carry data, hence
\begin{equation}
	C(T_{\tau})=\tfrac{T-T_{\tau}}{T}\log_{2}(1+\rho_{\mathrm{eff}}).
	\label{eq:C_tau}
\end{equation}
Using \eqref{eq:se2} in \eqref{eq:rho_eff} followed by \eqref{eq:C_tau} gives \eqref{eq:capcf}.

For the biased case in \eqref{eq:prior}, the same \ac{MMSE} estimator in \eqref{eq:cmb} gives,
\begin{equation*}
h_{0}-\hat h=\frac{h_{0}-\rho T_{\tau}\bar n-(1/\sDT)e_{\mathrm{DT}}}{1+\rho T_{\tau}+1/\sDT}-\frac{1/\sDT}{1+\rho T_{\tau}+1/\sDT}\,b,
\end{equation*}
Hence, we can show for the biased case that 
\begin{equation}
\sigma_{e,b}^{2}= \mathbb{E} \vert h_{0}-\hat h \vert^2 = \sigma_{e}^{2}+w_b^{2}|b|^{2},
\label{eq:se2b}
\end{equation}
where $w_b \triangleq \frac{1/\sDT}{1+\rho T_{\tau}+1/\sDT}$. Substituting \eqref{eq:se2b} for $\sigma_{e}^{2}$ in \eqref{eq:rho_eff}-\eqref{eq:C_tau} yields the biased rate in \eqref{eq:capcfb}.

\end{IEEEproof}
The uninformative \ac{DT} ($1/\sDT\to 0$) recovers
the pure training rate of \cite{hassibi2003howmuch}. On the other hand, a perfect \ac{DT}
($1/\sDT \to \infty$) drives the estimation error as $\sigma_{e}^{2} \to 0$, which in turn drives the effective \ac{SNR} towards the true one, namely $\rho_{\mathrm{eff}} \to \rho$,
and with $T_{\tau} \to 0$, $C \to \log_{2}(1+\rho)$ gives the genie rate. Next, we aim at answering the fundamental question of \emph{how much training is needed with a \ac{DT}} ?

\begin{theorem}[Optimal training length]\label{thm:opttrb}
For a biased \ac{DT} with unknown bias $b$, the rate \eqref{eq:capcfb} is maximized at
$T_{\tau}^{\star}=(u^{\star}-a)/\rho$, where $u^{\star}\in[a,\,a+\rho T]$ is the unique root of
\begin{equation}
\ln \frac{(1+\rho)\,u^{2}}{D(u)}
=(\rho T+a-u) \left(\frac{2}{u}-\frac{2u+\rho}{D(u)}\right),
\label{eq:optb}
\end{equation}
with $a\triangleq 1+\tfrac{1}{\sDT}$ and $D(u)\triangleq u^{2}+\rho u+\rho\,|b|^{2}/\sigma_{\mathrm{DT}}^{4}$.
\end{theorem}
\begin{IEEEproof}
Let $u\triangleq a+\rho T_{\tau}$, a monotone reparametrization of $T_{\tau}$ on
$[\,a,\,a+\rho T\,]$, so that $T_{\tau}=(u-a)/\rho$ and $\tfrac{T-T_{\tau}}{T}=\tfrac{\rho T+a-u}{\rho T}$.
With the biased error \eqref{eq:se2b}, one has $1+\rho\sigma_{e,b}^{2}=D(u)/u^{2}$, hence $1+\rho_{\mathrm{eff}}=\dfrac{(1+\rho)\,u^{2}}{D(u)}$ and
\begin{equation}
C^{b}(T_{\tau})=\frac{1}{\rho T\ln 2}\,(\rho T+a-u)\,\ln \frac{(1+\rho)\,u^{2}}{D(u)}.
\end{equation}
The factor $1/(\rho T\ln 2)$ is a positive constant in $u$, so maximizing $C^{b}$ is equivalent to
maximizing $g(u)\triangleq(\rho T+a-u)\,\ln\frac{(1+\rho)u^{2}}{D(u)}$. Using
$\frac{d}{du}\ln\frac{(1+\rho)u^{2}}{D(u)}=\frac{2}{u}-\frac{2u+\rho}{D(u)}$, i.e.
\begin{equation}
g'(u)=-\ln \frac{(1+\rho)\,u^{2}}{D(u)}
+(\rho T+a-u) \left(\frac{2}{u}-\frac{2u+\rho}{D(u)}\right),
\end{equation}
and setting $g'(u)=0$ yields \eqref{eq:optb}. The objective is unimodal on $[\,a,\,a+\rho T\,]$,
so the root $u^{\star}$ is unique and $T_{\tau}^{\star}=(u^{\star}-a)/\rho$.
\end{IEEEproof}
Even though \textbf{Lemma \ref{lem:cap}} quantifies the rate for a \ac{DT}-aided wireless system, one can also be interested in the relative gain the \ac{DT}-aided system can give in terms of rate when compared to a pilots-only case, where no \ac{DT} is available. In essence, the advantage a \ac{DT} brings in terms of rate can be expressed as
\begin{equation}
\Delta C(\rho,\sDT,b) \triangleq 
\underbrace{\max_{0\le T_{\tau}\le T} C^{b} \big(T_{\tau};\rho\big)}_{\text{\ac{DT}-aided}}
-
\underbrace{\max_{0\le T_{\tau}\le T} C \big(T_{\tau};\rho\big)}_{\text{pilots only}},
\label{eq:advantage}
\end{equation}
where $C^{b}$ is the biased \ac{DT}-conditioned rate \eqref{eq:capcfb} and the second term is
\eqref{eq:capcf} with $\sDT \to \infty$.
Indeed, \eqref{eq:advantage} measures a \ac{DT} rate advantage relative to having no \ac{DT}.
It is worth noting that $\Delta C(\rho) \to 0$ as $\rho \to \infty$, which means that the advantage is
maximized at a finite \ac{SNR}, where a \ac{DT} is most valuable when training is expensive, and is redundant as \ac{SNR} becomes higher and higher.

\section{Simulation Results}\label{sec:sim}

\begin{figure}[!t]
\centering
\includegraphics[width=\columnwidth]{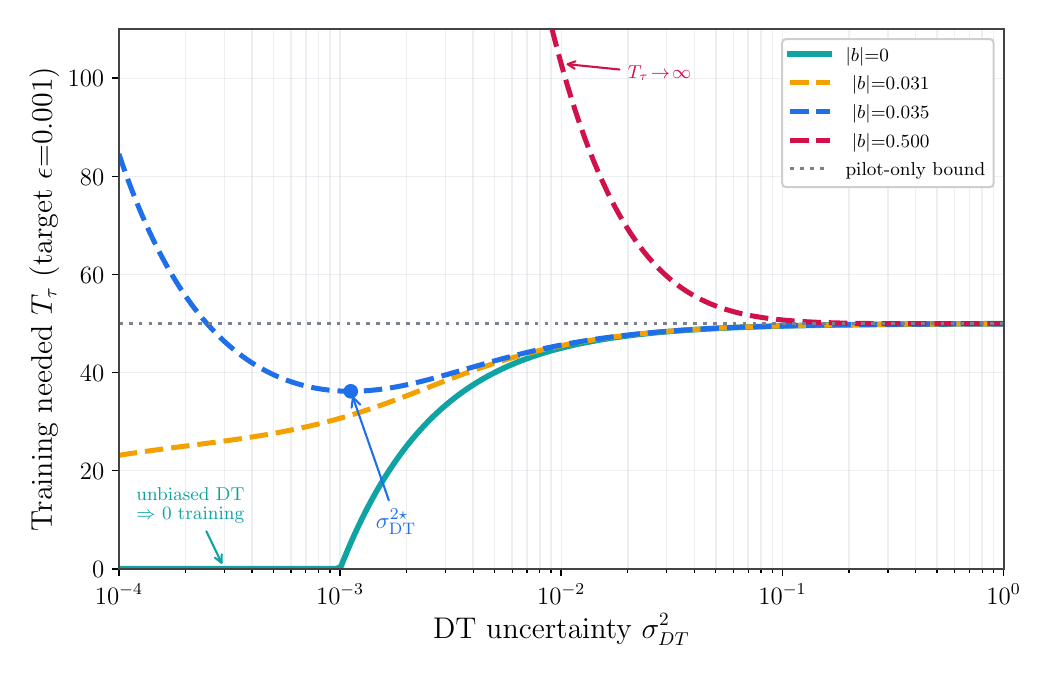}
\caption{Training length required in the \ac{PW} to reach a target \ac{MSE} versus the \ac{DT} uncertainty $\sDT$, for an unbiased and for several biased \acp{DT}.}
\label{fig:npilots}
\end{figure}

\begin{figure}[!t]
\centering
\includegraphics[width=\columnwidth]{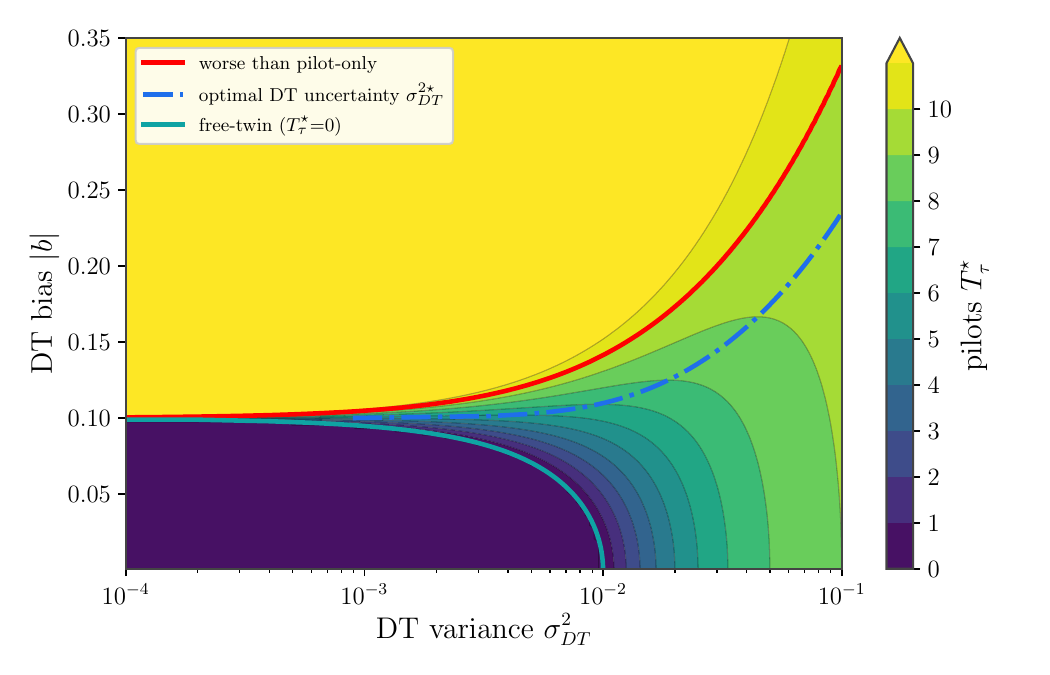}
\caption{Training length required in the \ac{PW} to reach a target \ac{MSE}, over the plane of \ac{DT} uncertainty and \ac{DT} bias.}
\label{fig:phase_diagram}
\end{figure}

\begin{figure}[!t]
\centering
\includegraphics[width=\columnwidth]{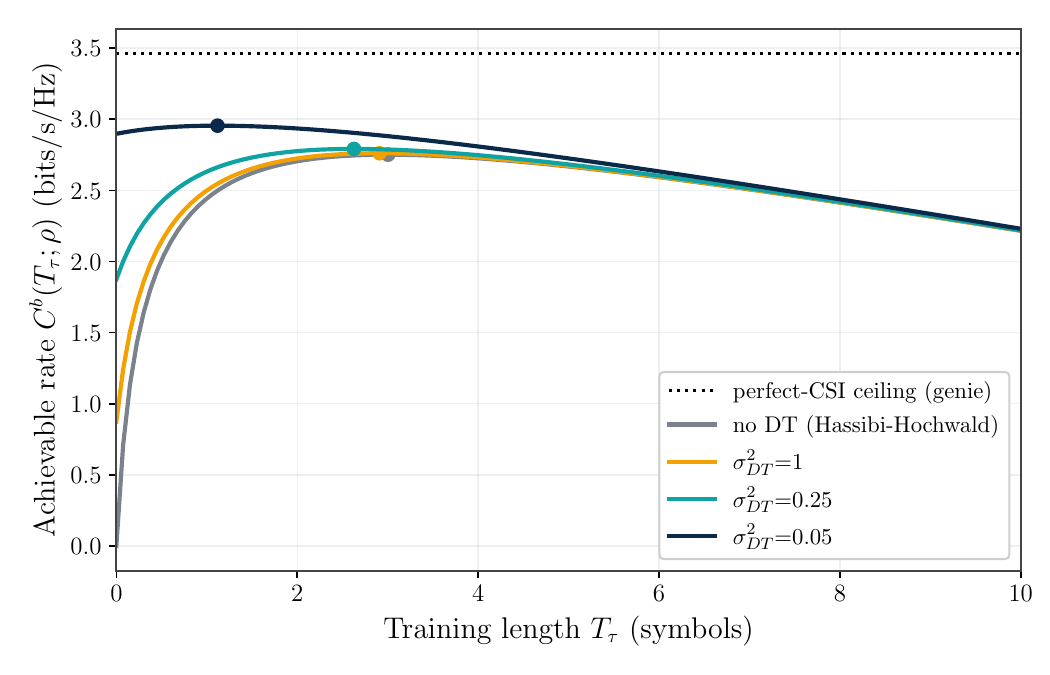}
\caption{Achievable rate of a \ac{DT}-aided system against the training length $T_{\tau}$ for several \ac{DT} variances, where the markers give the rate optimal training length.}
\label{fig:cap_train}
\end{figure}

\begin{figure}[!t]
\centering
\includegraphics[width=\columnwidth]{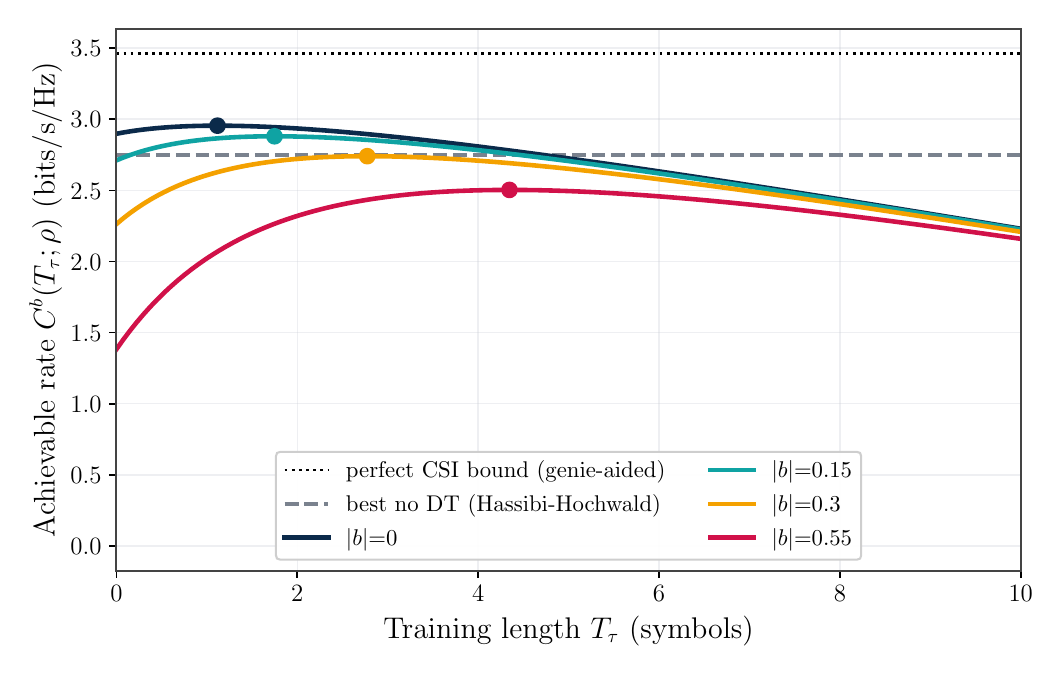}
\caption{Achievable rate of a \ac{DT}-aided system against the training length $T_{\tau}$ for several \ac{DT} biases at a fixed \ac{DT} variance of $\sDT = 0.05$.}
\label{fig:cap_train2}
\end{figure}

\begin{figure}[!t]
\centering
\includegraphics[width=\columnwidth]{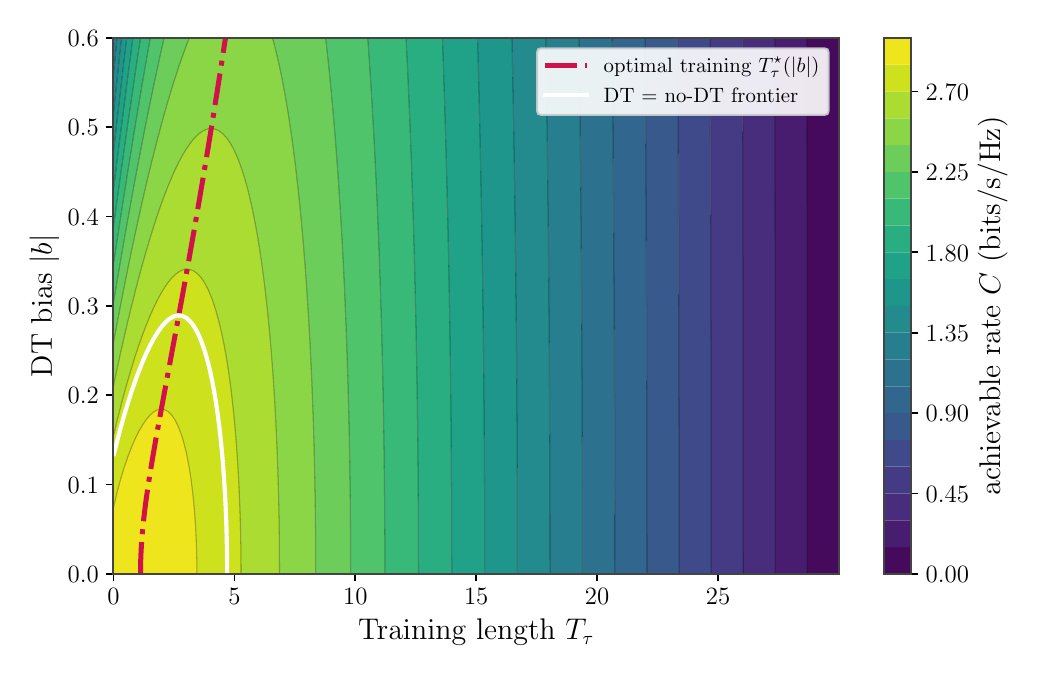}
\caption{Achievable rate of a \ac{DT}-aided system over the plane of training length and \ac{DT} bias.}
\label{fig:phase_C}
\end{figure}

\begin{figure}[!t]
\centering
\includegraphics[width=\columnwidth]{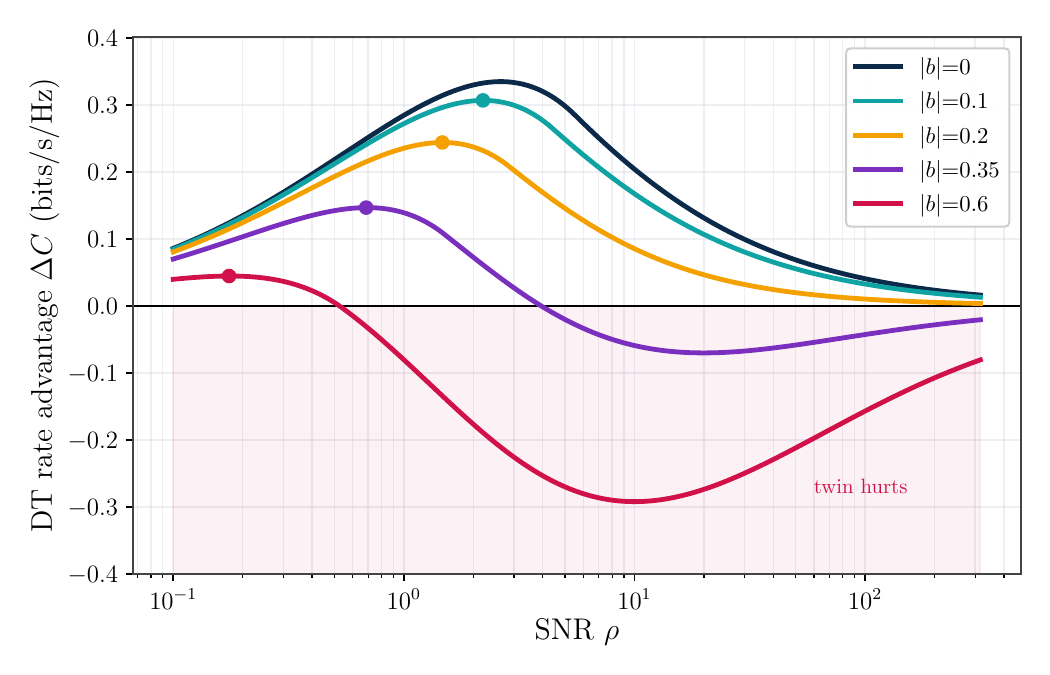}
\caption{Rate advantage of an optimally trained \ac{DT}-aided system over the optimally trained pilot-only system, against the operating \ac{SNR} $\rho$ for several \ac{DT} bias levels, with the markers giving the peak of each curve.}
\label{fig:fig12}
\end{figure}

We now validate the analysis and illustrate the corresponding operating regimes for the \ac{DT}-aided wireless system. 
The results are organized in two parts. We first study the channel estimation \ac{MSE}, i.e. the \ac{DT}-aided bound, the training length required to reach a target accuracy \eqref{eq:npilots}, the pilot equivalence of a \ac{DT}, and the mismatch thresholds beyond which a biased \ac{DT} is detrimental. The simulations depict quantitatively how much training is expected in a \ac{DT}-aided system as opposed to a no-\ac{DT} aided system. 
We then translate \ac{DT}-aided \ac{MSE} gains whenever suitable into achievable rate through the \ac{DT}-aided achievable rate of Lemma \ref{lem:cap} and plot the optimal training length of Theorem~\ref{thm:opttrb} together with the regimes in which a \ac{DT} provides gains and decreases the achievable rate.\\

In Fig.~\ref{fig:npilots}, we plot the required number of pilots needed to attain a given channel estimation \ac{MSE}.
In this plot, we set $(\epsilon,\rho) = (10^{-2},10 \ \mathrm{dB})$.
The \emph{"pilot-only"} bound represents a case whereby no \ac{DT} is available, hence training only happens in the \ac{PW} without being aided by a \ac{DT}, and therefore one would need to satisfy $\epsilon = \mathcal{E}_{\mathrm{PW}}$.
In the high uncertainty regime, namely as $\sDT$ grows large, we see that regardless of the bias, that all regimes converge towards the pilot-only bound, which indicates that the number of pilots required to attain a given channel estimation \ac{MSE} is dictated by the \ac{PW}. 
When the \ac{DT} is highly biased, e.g. $\vert b \vert = 0.5$, we see that the \ac{DT} always requires more pilots than the case of pilot-only. In addition, as the \ac{DT} becomes more descriptive (more certain) of the \ac{PW}, we see that the \ac{DT} would require more and more pilots to attain the required target channel estimation \ac{MSE}. For instance, if the $\sDT = 10^{-2}$, the \ac{DT} requires double the amount of pilots compared to the no-pilot case, when the bias is $\vert b \vert = 0.5$ and increases rapidly as the \ac{DT} becomes more certain.
Another biased case worth noting is when $\vert b \vert = 0.035$ (case where $\vert b \vert^2 > \epsilon$) where the minimum number of pilots required to attain a given channel estimation \ac{MSE} is a non-trivial point around $\sigma_{\mathrm{DT}}^{2\star} = 1.05 \times 10^{-3}$ (obtained by \eqref{eq:opt-sdt}) of \eqref{eq:opt-sdt} suggesting when the \ac{DT} error decreases beyond $\sDT \leq \sigma_{\mathrm{DT}}^{2\star}$, more weight $w$ is being placed on the \ac{DT} which would require more training is required to learn and undo such bias. 
When the bias goes below $\vert b \vert^2 \leq \epsilon$, then the optimal \ac{DT} uncertainty minimizing the required number of pilots is zero as explained below \eqref{eq:opt-sdt}. In that case, the bias is small enough that the required number of training needed simply increases with uncertainty to reach the pilot-only bound. Indeed, in such case, the \ac{DT} still bring benefits. For instance, at $\sDT = 10^{-3}$, the training time required is about $30$ slots, as opposed to $50$ in the pilots-only bound. 
For the unbiased \ac{DT}, as long as the target channel estimation \ac{MSE} is attained, namely $\sDT \leq \epsilon$, no training is required, otherwise, as the \ac{DT} becomes more uncertain the required number of pilots increases steadily to attain the number of pilots required for the pilot-only case.
We see that an accurate unbiased \ac{DT} requires no training and a biased one admits a best uncertainty beyond which more training is needed to due to its bias.\\

In Fig.~\ref{fig:phase_diagram}, we plot an operating-region map on the number of pilots required to attain a required target channel estimation \ac{MSE} $\epsilon$ at an operating \ac{SNR} for different values of bias and \ac{DT} uncertainty. In this plot, we set $(\epsilon,\rho) = (10^{-3},20 \mathrm{ dB})$. 
The boundary labeled as \emph{free-twin} is the boundary characterized by $\sDT+|b|^{2}\le\epsilon$. Under this condition, the \ac{DT} already meets the required target channel estimation \ac{MSE} $\epsilon = 10^{-3}$ and hence no training in the \ac{PW} is required. 
Another boundary worth noting is the \emph{optimal DT uncertainty} which arises when $\vert b \vert > \sDT$. It is the boundary upon which the number of pilots required for a given bias, dictated by \eqref{eq:opt-sdt}, is minimum.
In addition, the \emph{worse than pilot-only} boundary is dictated by \textbf{Proposition \ref{prop:neg}} and tells us that any $(b,\sDT)$ above that boundary means that incorporating the \ac{DT} with the \ac{PW} will give a worse \ac{MSE}, in the sense of $\mathcal{E}_{\mathrm{P+D}}^b>\mathcal{E}_{\mathrm{PW}}$.
In essence, we conclude that no training is needed below the free-twin boundary, whereas above the "worse than pilot-only" boundary the \ac{DT} costs more training than using no \ac{DT} at all, and the dash dotted locus gives the best \ac{DT} uncertainty for any given bias.\\

In Fig.~\ref{fig:cap_train}, we plot the achievable rate obtained for different training lengths, averaged over channel realizations for a fixed $T = 30$ and fixed \ac{SNR} set to $\rho = 10 \ \mathrm{dB}$. 
The figure depicts the fundamental tradeoff in wireless communications, where transmitting more pilots improves the channel estimate but cannibalizes the time available to actually send data, hence creating an optimal operation point for training length for maximum achievable rate. 
For a traditional system operating without a \ac{DT} (labeled as \emph{Hassibi-Hochwald} bound \cite{hassibi2003howmuch}), we observe that a significant chunk of the coherence block must be wasted on training just to reach optimal capacity. In particular, the optimal training length is $T_\tau \simeq 3$ symbols to maximize the achievable rate, which in this case attains about $2.75 \ \mathrm{bits/s/Hz}$.
However, as the \ac{DT}'s information becomes more precise (smaller $\sDT$), not only reduces the optimal training length, but also increases the achievable rate. For instance, when $\sDT = 0.05$, the optimal training length is about $T_\tau \simeq 1$ symbols, and achieves a capacity of $3 \ \mathrm{bits/s/Hz}$.
This means that an accurate \ac{DT} inherently reduces the physical overhead required for channel estimation, freeing up more symbols for data payload, which then leaves positive reverberation effects on the achievable rate. In other words, a more accurate \ac{DT} shortens the optimal training and increases the achievable rate towards the perfect-\ac{CSI} case, whereas an uninformative \ac{DT} coincides with the pilot only rate of Hassibi and Hochwald \cite{hassibi2003howmuch}.\\

In Fig.~\ref{fig:cap_train2}, similar to Fig.~\ref{fig:cap_train}, we plot again the achievable rate vs training lengths for the same simulation parameters, but we fix the uncertainty to $\sDT = 0.05$. The purpose is to illustrate the destructive impact of a biased \ac{DT} on both system capacity and training overhead.
In particular, Unlike an unbiased \ac{DT}, a \ac{DT} that is confident but biased actively misleads the channel estimator, which then forces the system to transmit significantly more pilots just to \emph{"out-vote"} the bad prior.
In fact, if the bias becomes very high, e.g. $\vert b \vert = 0.55$ the \ac{DT} becomes harmful to the system, not only from training length required (about $4$ symbols in this case), but also the achievable rate is reduced to $2.5 \ \mathrm{bits/s/Hz}$, as opposed to the case of $\vert b \vert = 0.3$, which achieves about $2.75 \ \mathrm{bits/s/Hz}$ with $3$ training symbols, which attains the best Hassibi-Hochwald bound (best in a sense of obtaining the maximum of that bound over $T_\tau$, namely the peak of the bound in Fig.~\ref{fig:cap_train}) of no \ac{DT}. 
One the bias reduces even more, the system goes beyond the no \ac{DT} bound, and starts to become closer to the genie-aided perfect \ac{CSI} bound. For instance, the unbiased \ac{DT} attains $3 \ \mathrm{bits/s/Hz}$ with only $1$ training symbol.\\ 

However, to understand the joint impact of training length and the biased introduced by the \ac{DT}, we plot a heatmap to provide a bird's-eye view of system capacity, as a function of training length $T_\tau$ and \ac{DT} bias in Fig.~\ref{fig:phase_C}.
The red "ridge" (obtained by solving \eqref{eq:optb}) representing optimal training bends towards the right as bias increases, which shows that a \ac{DT}-aided system must actively sacrifice more data symbols for physical pilots just to combat the \ac{DT}'s misleading prior.
The solid white contour line represents a cliff in performance, i.e. the threshold where the rate drops so low that the \ac{DT}-aided system would be better off ignoring the \ac{DT} entirely.\\

In Fig.~\ref{fig:fig12}, we fix as before $T=30$ symbols and $\sDT = 0.05$, where we plot the \ac{DT} rate gain $\Delta C(\rho)$ of \eqref{eq:advantage}, namely the difference between the optimally trained \ac{DT}-aided rate and the optimally trained pilot-only rate, as a function of the operating \ac{SNR} $\rho$ for several bias levels. The zero line represents the case whereby the \ac{DT} brings no benefit, i.e. the \ac{DT}-aided system attains exactly the pilot-only rate. For an unbiased \ac{DT}, we see that $\Delta C$ is positive throughout but decays monotonically towards zero as $\rho$ grows large, which indicates that at high \ac{SNR} the pilots become inexpensive enough that near-perfect \ac{CSI} is acquired by training alone, leaving little for the \ac{DT} to contribute. At the other extreme, as $\rho\to 0$, both rates vanish and so does $\Delta C$, so that the advantage is maximized at a finite \ac{SNR} of about $\rho \simeq 2.5$, indicating that a \ac{DT} is most valuable where training is expensive at that \ac{SNR}. 
On the other hand, when the \ac{DT} is biased, e.g.\ $|b|=0.35$, the advantage still rises at low \ac{SNR} but then crosses below zero over a mid-\ac{SNR} band, where the bias floor of \eqref{eq:se2b} dominates the estimation error before the pilots become cheap, which indicates that trusting the \ac{DT} is worse than ignoring it. As the bias grows further, e.g.\ $|b|=0.6$, this harmful band widens. For very large $\rho$, all curves converge towards zero, since sufficient training "washes out" the bias.\\

\section{Conclusion}\label{sec:concl}
In this paper, we addressed the question of how much training is needed with a \ac{DT} by modeling the DT as a complementary Gaussian measurement of the wireless channel, whose bias and variance follow from a Lindeberg-Feller central limit theorem over the DT's errors and are governed by the electromagnetic materials assigned to the environment. Moreover, using the \ac{DT} measurement along with physical world pilots via the best linear unbiased estimator, we derived a DT-aided Cram\'er-Rao bound, a pilot-equivalence law that converts DT fidelity into an equivalent number of pilots, and a bias mismatch threshold beyond which the \ac{DT} becomes destructive and unlearns the physical world. Casting the analysis as a block-fading achievable rate, we obtained the optimal training length in closed form and characterized when pilot training can be dispensed with altogether.

These results quantify the common claim that a more accurate \ac{DT} requires fewer pilots. We have shown that an incorrect material assumption within the \ac{DT} inflates its variance $\sDT$ by a factor of about $3.8$ through a ray tracing example, whereas the same error turns into a frequency-dependent bias once the path phases stay correlated across the deployment. Moreover, the number of pilots a \ac{DT} replaces decays inversely with the operating \ac{SNR}, hence the value of a \ac{DT}, in terms of achievable rate, is largest at a moderate \ac{SNR} and vanishes in both the low and high \ac{SNR} limits. Future directions include extending the framework to \acp{DT} whose bias is partially known to the estimator, to massive \ac{MIMO} and wideband millimeter wave channels, and to online \ac{DT} calibration that keeps the fidelity, and hence the pilot savings, high as the environment evolves. Future work can also include \ac{ISAC} tasks where the \ac{DT} can aid with.

\bibliographystyle{IEEEtran}
\bibliography{ref}

\end{document}